\documentclass[journal]{IEEEtran}

\usepackage{epstopdf}
\usepackage[draft]{hyperref} 
\usepackage{setspace}
\usepackage{amsmath}
\usepackage{amssymb}
\usepackage{amsthm}
\usepackage{multirow}
\usepackage{amsfonts}
\usepackage{color}
\usepackage{graphicx}
\usepackage{subfigure}  
\usepackage[]{algorithmicx}
\usepackage{algpseudocode,algorithm}
\usepackage{cite}
\usepackage{mathtools}
\usepackage{stmaryrd}
\usepackage[mathscr]{euscript}
\usepackage{mathrsfs}

\usepackage{soul}

\newtheorem{theorem}{Theorem}
\newtheorem{lemma}{Lemma}
\newtheorem{cor}{Corollary}
\newtheorem{defn}{Definition}

\begin{document}
		\title{Identifiability Conditions for Compressive Multichannel Blind Deconvolution}
	\author{Satish~Mulleti, {\it  Member, IEEE}, Kiryung Lee, {\it Senior Member, IEEE}, and Yonina C. Eldar, {\it Fellow, IEEE}
	
		\thanks{\scriptsize S. Mulleti and Y. C. Eldar are with the Faculty of Math and Computer Science, Weizmann Institute of Science, Israel. K. Lee is with the Department of Electrical and Computer Engineering at The Ohio State University. Email: mulleti.satish@gmail.com, lee.8763@osu.edu, yonina.eldar@weizmann.ac.il}
	\thanks{\scriptsize This project has received funding from the Benoziyo Endowment Fund for the Advancement of Science, Estate of Olga Klein -– Astrachanthe; European Union’s Horizon 2020 research and innovation program under grant No. 646804-ERC-COG-BNYQ; and the Israel Science Foundation under grant no. 0100101. K. Lee was supported in part by NSF under grant CCF 17-18771.}
}
	\markboth{DRAFT; }%
	{Shell \MakeLowercase{\textit{et al.}}: Bare Demo of IEEEtran.cls for Journals}
	\maketitle
	
	
	\begin{abstract}
	In applications such as multi-receiver radars and ultrasound array systems, the observed signals can often be modeled as a linear convolution of an unknown signal which represents the transmit pulse and sparse filters which describe the sparse target scenario. The problem of identifying the unknown signal and the sparse filters is a sparse multichannel blind deconvolution (MBD) problem and is in general ill-posed. In this paper, we consider the identifiability problem of sparse-MBD and show that, similar to compressive sensing, it is possible to identify the sparse filters from compressive measurements of the output sequences.  Specifically, we consider compressible measurements in the Fourier domain and derive identifiability conditions in a deterministic setup. Our main results demonstrate that $L$-sparse filters can be identified from $2L^2$ Fourier measurements from only two coprime channels. We also show that $2L$ measurements per channel are necessary. The sufficient condition sharpens as the number of channels increases asymptotically in the number of channels, it suffices to acquire on the order of $L$ Fourier samples per channel. We also propose a kernel-based sampling scheme that acquires Fourier measurements from a commensurate number of time samples. We discuss the gap between the sufficient and necessary conditions through numerical experiments including comparing practical reconstruction algorithms. The proposed compressive MBD results require fewer measurements and fewer channels for identifiability compared to previous results, which aids in building cost-effective receivers.

	\end{abstract}
	
	\IEEEpeerreviewmaketitle

	\section{Introduction}
In a wide range of applications, an unknown signal is observed through multiple channels. The output signal in each channel is given as the linear convolution of the unknown signal and the filter corresponding to the impulse response of the channel. The problem of identifying both the unknown signal and the filters is known as multichannel blind deconvolution (MBD). In general, this problem is ill-posed. It can be solved by imposing models on the source and the filters. In this paper, we consider the sparse-MBD problem, where the filters are assumed to be sparse. 

Sparse-MBD models arise in many practical applications such as radar imaging \cite{bajwa_radar, bar_radar}, seismic signal processing \cite{filho201seismic}, room impulse response modeling \cite{rip}, sonar imaging \cite{carter_sonar}, and ultrasound imaging \cite{eldar_sos,eldar_beamforming}, where a transmit signal is observed through multiple receivers after reflecting from sparsely located targets. The filters indicate the locations of the targets relative to the position of the receivers. Typically, the transmit signal is assumed to be known, however, in practice, it is often distorted while transmission and propagation \cite{bhandari_unknownsource}. Hence, the output signals from the receivers can be modeled within the sparse MBD framework.

In the aforementioned applications, the implementation cost is determined by the number of receivers (or equivalently the number of channels) and the computational cost is governed by the length of the output sequences. Hence, it is desirable to identify the MBD problem from a minimal number of channels and minimal number of samples per channel. We study the problem of identifying the sparse filters from fewer measurements of the output sequences compared to their ambient dimension, which we call as \emph{compressive MBD}. We show that compressive MBD is possible by combining the deterministic MBD approach developed for the non-sparse case \cite{xu_kilath} and the sparse signal identifiability results from compressive sensing framework \cite{eldar_2015sampling}. 

Any blind deconvolution linear measurements of the output signal suffers from shift and scaling ambiguity. Having redundant observations through multiple channels does not remove this fundamental ambiguity. Hence, the identifiability of the MBD problem is considered within fundamental ambiguity class. Even though the fundamental ambiguities are acceptable, in general, MBD is still an ill-posed problem and can be solved only by imposing additional conditions on the source and the filters. Several MBD results have been discussed in the literature where the source and the filters are assumed to have different structures in addition to the assumption that the filters have finite impulse responses (FIRs).
	
During the 90s several identifiability results and reconstruction algorithms for MBD have been presented, largely in the context of blind channel identification where the goal is to uniquely identify the filters \cite{ tong_fd, tong_td, moulines_det,gurelli1995evam, xu_kilath,  hua_wax, tong_survey}. The methods are classified as statistical \cite{tong_fd,tong_td,moulines_det,gurelli1995evam} or deterministic \cite{xu_kilath,  hua_wax}, depending on whether statistics of the source signal is used to identify the unknown filters (cf. \cite{tong_survey} for a comprehensive review of classical MBD results). In the statistical framework, it has been shown that an MBD problem is identifiable up to the fundamental ambiguities of scaling and shift if the source is zero-mean and white random process, and the filters are deterministic and coprime \cite{tong_survey}.
A set of sequences are coprime if their $z$-transform do not share any common zeros except the zeros at $z=0$. In this framework, first, the second-order statistics of the output sequences are estimated; the filters are then estimated from their statistics. The estimation accuracy of these approaches depends on how well the source statistics is known a priori and how accurately the second-order statistics are estimated from the available data. In applications where the source statistics may vary over time or is difficult to estimate from limited data, deterministic approaches are preferred.   

Xu et al. \cite{xu_kilath} developed a deterministic MBD approach for estimating FIR channels from their outputs to an unknown deterministic sequence where the filters are FIR with length $M_x$ (without any further sparsity condition). Starting from $3M_x$ truncated convolutive measurements  Xu et al. \cite{xu_kilath} showed that the filters  are uniquely identifiable under the following two conditions: i) the filters are coprime; ii) the linear complexity of the source, within the observation interval of the measurements, is greater than twice the order of the filters. The linear complexity of any sequence is a measure of its predictability and is given by the minimum number of exponentials which the signal consists of. In this paper, we show that the filters under the same FIR model are identifiable from $2M_x-1$ linear measurements with less restrictive condition on the source sequence compared to the linear complexity condition.
	
During the last decade, there has been renewed interest in MBD and, particularly in blind gain and phase calibration (BGPC) problems with sparsity and subspace constraints (e.g., \cite{yce_sensor,xia_li,balzano_nowak,cosse_mbd,wang_chi,lee_17,li2016optimal,ling_strohmer2,li2017blind,lee2018spectral,lee2018fast,ahmed_2018, gribonval_dl}). A BGPC is a bilinear inverse problem arising in a multi-sensor or multi-receiver system where the objective is to determine the unknown gains and phases of the sensing system as well as the unknown observed signals from the sensors. It can be shown that the Fourier-domain formulation of an MBD problem is a special case of BGPC. In this case, the unknown gains and phases are given by the Fourier coefficients of the common source sequence and the unknown observed signals are Fourier transforms of the filters. Recent results analyzed the case of uniform samples in the Fourier domain under the assumption that the sparse filters are random or generic \cite{balzano_nowak, wang_chi, cosse_mbd, lee_17}. In these works, it is assumed that the output sequences are obtained by the circular convolution between the source and the sparse filters. Assuming that all the output samples are available and the filters are random and sparse, identifiability results have been derived in terms of a sufficient number of channels. Specifically, in \cite{wang_chi} and \cite{lee_17}, it is shown that the sparse MBD problem is uniquely identifiable provided that the filter coefficients are modeled as independent and identically distributed Bernoulli-Gaussian random variables and all $M$ output samples are available from $N = \mathcal{O}(M \log^4 M)$ channels. In addition, they have shown that sparse MBD and the corresponding BGPC problem can be solved by a practical algorithm \cite{wang_chi, li2017blind}.

BGPC can also be posed as a blind dictionary calibration (BDC) problem where the goal is to recover the calibration weights for a known dictionary together with the sparse vectors. Gribonval et al. \cite{gribonval_dl} showed that a BDC problem can be posed as a convex optimization problem and a solution can be achieved by using off-the-shelf optimization solvers. However, identifiability results are not derived in \cite{gribonval_dl}. We will apply the algorithm in \cite{li2017blind} and an alternate minimization approach for the problem formulation in \cite{gribonval_dl} to compressive MBD in Section \ref{sec:simulations}.

There are three major limitations in applying the existing results \cite{lee_17, wang_chi, cosse_mbd} to our setting.
\begin{enumerate}
    \item The sparsity was introduced to solve sparse-MBD from fully observed output sequences. Their results do not apply to the case where it is enforced or desired to identify the signals from partial observations. 

    \item Their identifiability result has been derived with the number of channels increasing in the signal length. This sufficient condition is conservative in the sense that sparse-MBD can be solved empirically with fewer channels, that is, only two channels. 
	    
    \item Their analysis of sparse-MBD assumed that the filters follow certain stochastic models, which are not relevant to practical applications of our interest. Therefore their results do not apply due to the model mismatch. 
	
\end{enumerate}

Our main results, summarized below, overcome the above limitations in the existing results on sparse MBD. 

We present a set of identifiability results on sparse MBD that apply uniformly to any instance satisfying given constraints. In other words, unlike some of the recent results relying on certain stochastic models, our identifability results are purely deterministic. We show that the sparsity constraint enables compressive-MBD similar to compressive sensing. Specifically, it is possible to recover the filters and the source from a small number of Fourier measurements. 

In the sparse and deterministic setup, we consider two sub-problems. The first is to uniquely identify only the filters. Such a problem is useful in the radar and sonar applications where the filters contain the information about the targets and the source need not be identified. We combine the ideas of Xu et al. \cite{xu_kilath} and compressive sensing to derive identifiability results. By applying a cross-convolution approach, the identifiability problem reduces to the recovery of a superposition of convolutions of sparse filters from their partial Fourier measurements. Our main results follow by applying the full spark property of partial Fourier matrices and the coprimeness condition of the filters. We show that to identify the filters, two channels ($N=2$) are sufficient. When $N=2$, taking $2L^2$ Fourier samples per channel are sufficient. Furthermore, we show that the problem is not uniquely identifiable from fewer than $2L$ measurements per channel. For $N \gg 2$, we demonstrate a gain and establish that on average (averaged over the number of channels) we need order $L$ measurements per channel. 

Second, we consider the simultaneous identification of both the source and the filters. Using the results of filter identifiability, we propose a pairwise measurement strategy where we consider different Fourier measurements from each unique pair of channels. Starting from $N \geq 2$ channels, we show that $2L^2$ Fourier measurements are sufficient from $\max \{\frac{M_s-L^2-1}{L^2-1}, 2\}$ channels and $2L$ from the rest. Here $M_s$ is length of the source. We discuss practical algorithms to identify the source and the filters for the two-channel case \cite{li2017blind, gribonval_dl}. By applying these algorithms, we discuss the gap between the necessary and sufficient conditions through simulations. 
	
Our main results are derived by assuming that only partial Fourier measurements of the output sequences are available. However, in certain applications, the output sequences can be measured only in the time domain. We propose a sampling-kernel based technique that computes partial Fourier measurements without accessing all time samples. Then we obtain the analogous identifiability result from time-domain samples. 

We also specialize our results to the non-sparse FIR case. 
Our frequency-domain approach requires fewer measurements compared to the classical time-domain approach by Xu et al. \cite{xu_kilath}. Moreover, our approach is guaranteed when the Fourier transform of the source signal does not vanish at the observed frequencies, which is a milder condition than the analogous condition on the linear complexity of the source \cite{xu_kilath}. 

The rest of the paper is organized as follows. In the next section, we present the problem formulation along with relevant mathematical preliminaries. In Section \ref{sec:compressive_measurements} we show how to achieve compressive Fourier measurements from a finite set of time-domain samples. Identifiability results for compressive MBD are discussed in Section \ref{sec:results}. In Section \ref{sec:comparision}, we present a detailed comparison of the proposed results with the recent sparse and classical non-sparse results. Simulations are shown in Section \ref{sec:simulations} followed by the proof of the main results in Section \ref{sec:proof}. 

Throughout the paper, we use the following notations. For a positive integer $M$, let $[M]$ denote the set $\{0,1,\dots,M-1\}$. For a sequence $x$, its support denoted by $\mathrm{supp}(x)$ is defined by $\{k \in \mathbb{Z} ~|~ x[k] \neq 0\}$. The $\ell_0$ pseudo-norm of $x$, denoted by $\|x\|_0$, counts the number of its nonzero elements. the $z$-transform and the discrete-time Fourier transform (DTFT) of $x$ are denoted by $X(z)$ and $X(e^{\mathrm{j}\omega})$ respectively.

\section{Compressive MBD from Fourier Measurements}
In this section, we formulate the compressive MBD problem and discuss the assumptions made to deriving the main results. 

\subsection{Problem Statement}
\label{sec:problem_statement}

Let $y_1,\dots,y_N$ denote multichannel output sequences from a common source sequence $s$. Let $x_1,\dots,x_N$ denote the filters corresponding to the impulse responses of the channels. Then 
\begin{equation}
\label{eq:genmdl_yn}
    y_n = s \mathop{\ast} x_n, \quad n = 1,\dots,N,
\end{equation}
where $\ast$ denotes the linear convolution.  
The MBD problem is to identify $s$ and $\{x_n\}_{n=1}^N$ from the output sequences $\{y_n\}_{n=1}^N$. 

Let 
\begin{equation}
\label{eq:ambiguitu_class}
\begin{aligned}
    \mathcal{C}\left(s,\{x_n\}_{n=1}^N\right) =& \Big\{  \left(\alpha^{-1} \mathcal{S}_{-m}(s), \left\{\alpha \mathcal{S}_m(x_n)\right\}_{n=1}^N \right) \big| \\
    & \alpha \neq 0, ~ m \in \mathbb{Z} \Big\}
\end{aligned}
\end{equation}
denote the orbit of $(s, \{x_n\}_{n=1}^N)$ by the actions of shift and scaling, where $\mathcal{S}_m$ denotes a shift operator that maps a sequence to its shifted version by $m$ samples. Then any element in $\mathcal{C}(s,\{x_n\}_{n=1}^N)$ generates the same output sequences. Therefore, we aim to identify $(s, \{x_n\}_{n=1}^N)$ up to the \emph{fundamental-ambiguity class} given by \eqref{eq:ambiguitu_class}, that is, to find any element in $\mathcal{C}(s,\{x_n\}_{n=1}^N)$ that satisfies \eqref{eq:genmdl_yn}. Unique identification hereafter will be referred to as this case. 

When all time-samples of $y_n$ are available, one can uniquely identify the filters (resp.  both the source and the filters) by the method by Xu et al. \cite{xu_kilath} (resp. recent sparse MBD methods by \cite{lee_17, wang_chi, cosse_mbd}) provided that the assumed conditions on $s$ and $\{x_n\}_{n=1}^N$ therein are satisfied (cf. Section \ref{sec:comparision} for details). 

The main question of our interest is whether one can deconvolve $s$ and $\{x_n\}_{n-1}^N$ from the compressive measurements of $\{y_n\}_{n=1}^N$, which will be referred to as compressive MBD.

We specifically consider partial Fourier measurements given as the DTFT of $y_n$ at selected frequencies, that is, the set of linear measurements is written as  
\begin{align}
\{ Y_n(e^{\mathrm{j}k\omega_0}) ~|~ k \in \mathcal{K}_n, ~ n=1,\dots,N \},
\label{eq:new_mes0}
\end{align}
where $\mathcal{K}_n$ denotes the index set $ \{m_{n,1}, m_{n,2}, \dots, m_{n,K_n}\} \subset \mathbb{Z}$ for $n=1,\dots,N$. 
We are particularly interested in the setting when $|\mathcal{K}_n| \ll M$ where $M$ is length of $y_n$.

Our choice of the partial Fourier measurements is motivated by the following two reasons. First, in the Fourier-domain, the measurements in \eqref{eq:new_mes0} are written as the product of the Fourier transforms of the source and the filters, that is, 
\begin{align}
\hspace{-.09in}Y_n(e^{\mathrm{j}k\omega_0}) = S(e^{\mathrm{j}k\omega_0}) X_n(e^{\mathrm{j}k\omega_0}), ~ k \in \mathcal{K}_n, ~ n=1,\dots,N.
\label{eq:new_mes}
\end{align}
The entrywise product form in \eqref{eq:new_mes} allows the flexibility to design sampling patterns. For any $n_1 \neq n_2$ and $k \in \mathcal{K}_{n_1} \cap \mathcal{K}_{n_2}$, we have
\[
Y_{n_2}(e^{\mathrm{j}k\omega_0}) X_{n_1}(e^{\mathrm{j}k\omega_0})
= 
Y_{n_1}(e^{\mathrm{j}k\omega_0}) X_{n_2}(e^{\mathrm{j}k\omega_0}),
\]
which enables to apply the cross-convolution approach by Xu et al. \cite{xu_kilath} even with sampling in the Fourier domain for any $\{\mathcal{K}_n\}_{n=1}^N$. 
Second, due to the uncertainty principle, the Fourier transforms of the filters are well spread in the Fourier domain as they are sparse in the time domain. This helps recover the filters from fewer Fourier coefficients.

In order to uniquely identify the solution to compressive MBD, we impose the following structural assumptions on the filters $\{x_n\}_{n=1}^N$ and the source $s$: 
\begin{description}
\item[(A1)] \textbf{Sparse filters:} $\|x_n\|_0 \leq L$ and $\text{supp}\{x_n\} \subset [M_x]$ for $n = 1, \dots, N$. 
\item[(A2)] \textbf{Finite-length source:} $s$ is supported within $[M_s]$.
\item[(A3)] \textbf{Coprime filters:} $X_1(z),\dots,X_N(z)$ do not share any common zeros except at $z=0$. 
\item[(A4)] \textbf{Non-vanishing source:} $S(e^{\mathrm{j}k\omega_0}) \neq 0$ for all $k \in \cup_{n=1}^N \mathcal{K}_n$. 
\item[(A5)] \textbf{Universal sampling:} The sampling interval $\omega_0$ and index sets $\{\mathcal{K}_n\}_{n=1}^N$ form the universal sampling sets. Such sets are defined in Section \ref{sec:universal_set}.
\end{description}
A few remarks on these assumptions are in order. 
\begin{itemize}
    \item (A1) and (A2) imply that each $y_n$ is supported on $[M]$, where $M = M_x+M_s-1$.
    \item (A2) is not necessary if one only concerns the identification of the filters. 
    \item (A3) is a necessary condition for unique identifiability without (A1) and (A2). See Section \ref{sec:coprime} for more details.
    \item (A4) avoids the case where the Fourier measurements at a sampled frequency are zero for all channels. 
    \item (A5) enables to solve sparse MBD from compressive Fourier measurements.
\end{itemize}

Unique identification in compressive MBD is then defined as follows.
\begin{defn}[Identifiability of Compressive MBD]
\label{def:identifiability}
Compressive MBD is uniquely identifiable if any feasible solution $(\hat{s}, \{\hat{x}_n\}_{n=1}^N)$, which satisfies (A1)-(A4) and is consistent with the measurements in \eqref{eq:new_mes0}, belongs to the fundamental ambiguity class $\mathcal{C}\left(s, \{x_n\}_{n=1}^N\right)$ of the ground-truth signals $(s,\{x_n\}_{n=1}^N)$ defined in \eqref{eq:ambiguitu_class}.
\end{defn}

Our objective is to derive necessary and sufficient conditions on the index sets $\{\mathcal{K}_n\}_{n=1}^N$ such that either only the filters $\{x_n\}_{n=1}^N$ or both the source $s$ and the filters are uniquely identifiable according to Definition~\ref{def:identifiability} under (A5). 

Simultaneous identification of both the source and the filters requires extra conditions, which go beyond universal sampling (see Section \ref{ssec:source_recovery}). We first derive conditions for the unique identification of the filters followed by those extra conditions for the identification of the source signal. Note that recovery of the filters is equivalent to the simultaneous recovery of both the filters and the DTFT of the source signal at the observed frequencies. Therefore partial identifiability is as follows. 

\begin{defn}[Partial Identifiability Only for Filters]
\label{def:partial_identifiability}
Compressive MBD is partially identifiable if for any feasible solution $(\hat{s}, \{\hat{x}_n\}_{n=1}^N)$, there exists $\tilde{s}$ such that
\begin{enumerate}
    \item $\tilde{S}(e^{\mathrm{j}k\omega_0}) = \hat{S}(e^{\mathrm{j}k\omega_0})$ for $k \in \cup_{n=1}^N \mathcal{K}_n$. 
    \item $(\tilde{s},\{\hat{x}_n\}_{n=1}^N) \in \mathcal{C}\left(s, \{x_n\}_{n=1}^N\right)$.
\end{enumerate}
\end{defn}

\subsection{Universal Sets} \label{sec:universal_set}
Universal sampling sets have been introduced for compressed sensing from partial Fourier measurements (e.g., see \cite[Def. 14.1]{eldar_2015sampling}). The partial Fourier measurement matrix corresponding to the $n$th channel measurements in \eqref{eq:new_mes0} is given as a $|\mathcal{K}_n| \times \bar{M}$ Vandermonde matrix $\mathbf{V}_n$ whose $(k,m)$th element is given as $e^{\mathrm{j}km\omega_0}$. In our settings, it is satisfied that $\bar{M}\geq |\mathcal{K}_n|$, where $\bar{M} = \max\{2M_x-1,M_s\}$ (See Section \ref{ssec:source_recovery} for details). To avoid aliasing, $\omega_0$ is chosen such that the elements of the set $\{e^{m\omega_0}\}_{m=0}^{\bar{M}-1}$ are distinct. For each $n \in \{1,\dots,N\}$, the index set $\mathcal{K}_n$ is called \emph{universal} if every submatrix of $\mathbf{V}_n$ obtained by taking $|\mathcal{K}_n|$ columns has full rank \cite[Def. 14.1]{eldar_2015sampling}, that is, $\mathbf{V}_n$ has full spark \cite{eldar_cs_book}. Note that the universal sets depend on the frequency interval $\omega_0$.

For example, if $\mathcal{K}_n$ is a set of consecutive integers and $\omega_0 = 2\pi/\bar{M}$ then each index set $\mathcal{K}_n$ is universal. Various alternative constructions of universal sets have been studied (e.g.,  \cite{tao2005uncertainty,mishali2009blind,alexeev_spark,achanta_spark}, also see \cite{eldar_2015sampling}).

\subsection{Coprimeness of the Filters}
\label{sec:coprime}
In an MBD framework, unless any further restriction is imposed on the supports of the source and the filters, the coprime condition on the filters is a necessary for the unique identification of the solution. To elaborate, let the filters $\{x_n\}_{n=1}^N$ share nontrivial common zeros in the $z$-domain. Specifically, each filter can be decomposed as
\begin{align}
x_n = h_0*\hat{x}_n,
\label{eq:coprime}
\end{align}
where the sequence $h_0$ contains the common zeros except at $z=0$ and $\hat{x}_n$ is the novel factor. As a result, the outputs $\{y_n = s*x_n\}_{n=1}^N$ can also be decomposed as 
\begin{equation}
\label{eq:coprime2}
\begin{aligned}
y_n = s*x_n
= s*h_0*\hat{x}_n = \hat{s}*\hat{x}_n,
\end{aligned}
\end{equation}
where $\hat{s} =s*h_0$ and $\{\hat{x}_n\}_{n=1}^N$ provide an alternative solution that produces the same outputs. Hence, without any assumptions on the source and the filters, coprimeness is a necessary condition. With (A1) and (A2), it is no longer a necessary condition as the alternative solution may not satisfy these assumptions. However, we keep the coprimeness assumption to our settings to derive the identifiability conditions.

\section{Compressive MBD from Time Domain Measurements}
\label{sec:compressive_measurements}

In this section we propose a method that acquires the compressive Fourier measurements of an FIR sequence without explicitly observing the entire sequence. The number of time samples can be as small as the number of Fourier measurements. Our approach is inspired by the kernel-based sampling and reconstruction approach for finite-rate-of-innovation (FRI) signals \cite{eldar_sos, mulleti_kernal}. It has been shown that Fourier measurements of FRI signals can be computed from time samples at a sub-Nyquist rate by applying a suitable sampling kernel. Then the parameters of FRI signals can be computed from Fourier measurements on a grid. 
	
\begin{figure}
\centering
\includegraphics[width = 2.8in]{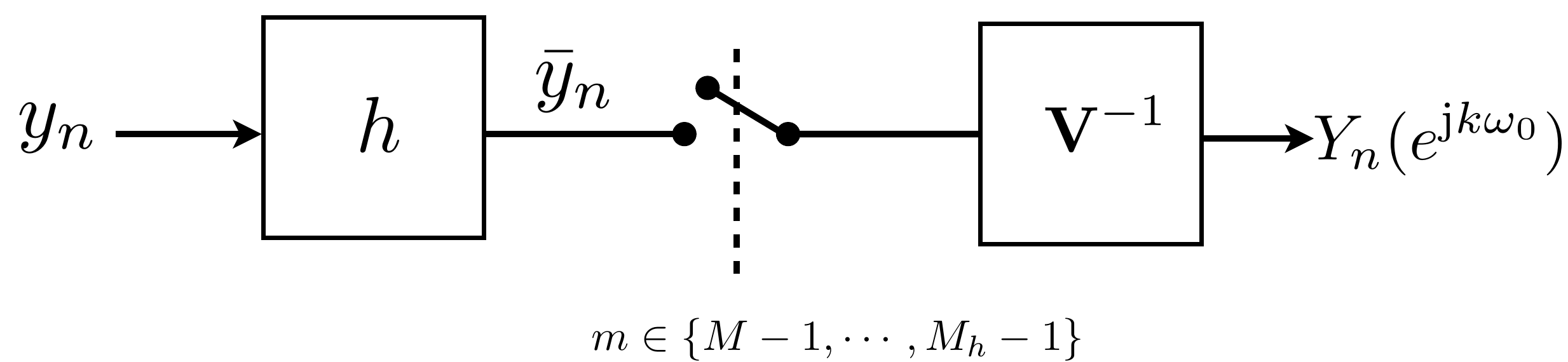}
\caption{Compressive measurements of Fourier samples of $y_n$ by using discrete-SOS filter $h$: the sequence $y_n$ of dimension of $M$ is passed through an FIR kernel of length $M_h$ defined as in \eqref{eq:sos}. The output sequence $\bar{y}_n$ has length $M_h+M-1$ out of which few measurements are taken by closing the switch at sample indices $m \in \{M-1, \dots, M_h-1\}$. From the truncated samples $Y_n(e^{\mathrm{j}k\omega_0})$ is computed by inverting a linear system of equations governed by matrix $\mathbf{V}$ considered in Section \ref{sec:universal_set}}.
\label{fig:sos}
\end{figure}
	
Let $y$ be an FIR signal supported on $[M]$ and $\mathcal{K} \subset \mathbb{Z}$ be a finite set. The DTFT coefficients $Y(e^{\mathrm{j}k\omega_0})$ of $y$ for $k \in \mathcal{K}$ are computed from $|\mathcal{K}|$ consecutive time samples of $y$ with an appropriate sampling kernel. 
	
Let $h$ be an FIR filter supported on $[M_h]$, where $M_h > M$, which satisfies
\begin{align}
h[m] = \sum_{k\in \mathcal{K}} e^{\mathrm{j}k\omega_0 m}, \quad m \in [M_h].
\label{eq:sos}
\end{align}
The filter $h$ in $\eqref{eq:sos}$ is a discrete-time analog of the sum-of-sincs (SOS) filter proposed by Tur et al. \cite{eldar_sos}. The filtered version of $y$ by $h$, denoted $\bar{y}$, satisfies
\begin{equation}
\label{eq:sos1}
\bar{y}[m] = \sum_{p=0}^{M-1} y[p] h[m-p] =\sum_{k \in \mathcal{K}}e^{\mathrm{j}k \omega_0 m}\, Y(e^{\mathrm{j}k \omega_0 })
\end{equation}
for $M-1\leq m \leq M_h-1$. Then the vectors respectively containing $\{Y(e^{\mathrm{j}k \omega_0})\}_{k \in \mathcal{K}}$ and $\{\bar{y}[m]\}_{m=M-1}^{M_h-1}$ are related through a Vandermonde matrix $\mathbf{V}$ of size $(M_h-M+1) \times |\mathcal{K}|$ with its $(m,k)$th element given as $e^{\mathrm{j}mk\omega_0}$. Therefore if $\omega_0$ and $\mathcal{K}$ form the universal set then the condition $M_h\geq M+|\mathcal{K}|-1$ ensures that $\{Y(e^{\mathrm{j}k \omega_0})\}_{k \in \mathcal{K}}$ are computed uniquely from $\{\bar{y}[m]\}_{m=M-1}^{M_h-1}$. In other words, $|\mathcal{K}|$ observations of the $\bar{y}$ are sufficient to compute $|\mathcal{K}|$ Fourier measurements of $y$ without observing the entire sequence. 

A schematic of the sampling mechanism is shown in Fig.~\ref{fig:sos}. The switch is closed for the samples at $m \in \{M-1, \dots, M_h-1\}$. If we further assume that $\mathcal{K}$ is a universal set, any $|\mathcal{K}|$ time samples of $\bar{y}$ from $[M-1,M_h-1]$ will make the resulting matrix $\mathbf{V}$ full rank. This enables optimizing the sampling pattern to improve the condition number of $\mathbf{V}$.

\section{Identifiability of Compressive MBD}
\label{sec:results}
We consider a two-step approach to the compressive MBD problem. The first step identifies only the filters corresponding to the impulse responses of the channels, similarly to blind channel estimation in communications (e.g., \cite{hua_wax, xu_kilath}). Once the filters are identified, then the second step reconstructs the common input source to the channels. 
	
\subsection{Identifying Sparse Filters}
\label{ssec:filter}
We identify the filters from the Fourier measurements in \eqref{eq:new_mes} under the assumptions (A1) to (A5) except (A2). The following theorem presents the result in the two-channel case.	
 
\begin{theorem}[Partial Identifiability of Compressive MBD]
\label{theorem1}
{Suppose that (A1), (A3), (A4), and (A5) hold with $\mathcal{K}_1 = \mathcal{K}_2 = \mathcal{K}$ and $L< \sqrt{M_x}$.
\begin{enumerate}
    \item If $|\mathcal{K}| \geq 2L^2$, then compressive MBD is partially identifiable from the Fourier measurements according to Definition~\ref{def:partial_identifiability}.
    \item If $|\mathcal{K}| < 2L$, then compressive MBD is not partially identifiable. 
\end{enumerate}
}
\end{theorem}

Recall that $\{S(e^{\mathrm{j}\omega_0k})\}_{k \in \mathcal{K}}$, $x_1$, and $x_2$ were arbitrary in Definition~\ref{def:partial_identifiability}. Therefore, the partial identifiability in Theorem~\ref{theorem1} implies that for any instance within the assumed model, the filters are uniquely identified up to the ambiguity class. When it is not partially identifiable, there exists an instance where the filters are not uniquely determined.

Next, by combining the Fourier domain identifiability results of Theorem~\ref{theorem1} with the kernel-based measurement scheme proposed in Section~\ref{sec:compressive_measurements}, we obtain the following corollary, which provides the analogous results for compressive MBD from time-domain measurements.
	
\begin{cor}[Time-Domain Compressive MBD]
\label{cor:time}
Suppose that the hypotheses of Theorem~\ref{theorem1} and (A2) hold. Let $h$ be an FIR filter of length $|\mathcal{K}|$ with impulse response defined in \eqref{eq:sos}. Then compressive MBD is partially identifiable from $|\mathcal{K}|$ consecutive time-samples of $y_1 \ast h$ and $y_2 \ast h$, where the samples are indexed by the set $\{M-1,\dots, M+|\mathcal{K}|-2\}$, if $|\mathcal{K}|\geq 2L^2$. On the other hand, compressive MBD is not partially identifiable if $|\mathcal{K}|<2L$. 
\end{cor}

The time-domain results are similar to the result in the context of FRI signal sampling where the number of Fourier coefficients needed for the identification governs the desired number of time samples and the sampling rate \cite{eldar_sos}.
	
The proof of Theorem~\ref{theorem1} directly follows from the results on the following feasibility problem:
\begin{equation}
\label{eq:optNEW}
\begin{array}{lll}
& \text{find} & \{\tilde{S}(e^{\mathrm{j}\omega_0k})\}_{k \in \cup_n \mathcal{K}_n}\quad \text{and}\quad ( \tilde{x}_1, \tilde{x}_2, \dots, \tilde{x}_N)  \\
& \text{s.t.} &  Y_n(e^{\mathrm{j}\omega_0k}) = \tilde{S}(e^{\mathrm{j}\omega_0k}) \tilde{X}_n(e^{\mathrm{j}\omega_0k}), ~ k \in \mathcal{K}_n, ~ \forall n,  \\
& & \text{(A1), (A3), (A4) and (A5) are satisfied,} 
\end{array}
\end{equation}
where the tilde is used to distinguish the variables of the problem from the corresponding ground-truth signals. If the solution to \eqref{eq:optNEW} is unique up to a scaling and shift ambiguity, then $(\{S(e^{\mathrm{j}\omega_0k})\}_{k \in\bigcup_{n=1}^N \mathcal{K}_n}, \{x_n\}_{n=1}^N)$ is uniquely identifiable, up to a shift and scaling ambiguity, from the measurements in \eqref{eq:new_mes}. The following lemma, whose proof is deferred to Section \ref{sec:proof}, provides necessary and sufficient conditions for the uniqueness of the feasibility problem in \eqref{eq:optNEW} when $N=2$. 

\begin{lemma}
\label{lemma1}
{Let $\omega_0$, $\mathcal{K}_1$, $\mathcal{K}_2$, and $\mathcal{K}$ be as in Theorem~\ref{theorem1}. If $|\mathcal{K}| \geq \min \{2L^2, 2M_x-1\}$, then the feasibility problem in \eqref{eq:optNEW} has a unique solution. On the other hand, if $|\mathcal{K}|<2L$, then the problem in \eqref{eq:optNEW} is not uniquely identifiable.}
\end{lemma}
	
The ground-truth signals $(\{S(e^{\mathrm{j}\omega_0k})\}_{k \in \mathcal{K}_1 \bigcup \mathcal{K}_2}, \{x_{1}, x_{2}\})$ are feasible to \eqref{eq:optNEW}. Let
\begin{align}
q = x_1 \ast \hat{x}_2 - x_2 \ast \hat{x}_1,
\label{eq:q}
\end{align}
where $(\{\hat{S}(e^{\mathrm{j}\omega_0k})\}_{k \in \mathcal{K}}, \hat{x}_1, \hat{x}_2)$ is another feasible solution to \eqref{eq:optNEW}. In the proof of Lemma~\ref{lemma1}, we have shown that the number of Fourier measurements for the unique identification is determined as the worst-case $\|q\|_0$ maximized over all feasible $(\hat{x}_1,\hat{x}_2)$. Since $x_1, x_2, \hat{x}_1$, and $\hat{x}_2$ are $L$-sparse vectors with support over $[M_x]$, in general, the worst case support of $q$ is $\min \{2L^2, 2M_x-1\}$. For high sparse signals, that is, when $L \ll M_x$ or $L<\sqrt{M_x}$, we have that $\|q\|_0 = 2L^2$. Therefore by modifying constraints on $\tilde{x}_1$ and $\tilde{x}_2$ in \eqref{eq:optNEW}, we obtain similar results in a different scenario as an immediate corollary. 

\begin{cor}[Sufficient Conditions for General Sparsity Case]
\label{lemma_new}
{Let $\omega_0$, $\mathcal{K}_1$, $\mathcal{K}_2$, and $\mathcal{K}$ be as in Theorem~\ref{theorem1}. Let $L\leq M_x$ without the restriction $L<\sqrt{M_x}$. If $|\mathcal{K}| \geq \min \{2L^2, 2M_x-1\}$, then the feasibility problem in \eqref{eq:optNEW} has a unique solution. On the other hand, if $|\mathcal{K}|<2L$ the problem in \eqref{eq:optNEW} is not uniquely identifiable.}
\end{cor}
	
Without assuming that $x_1$ and $x_2$ are sparse, the worst-case $\|q\|_0$ becomes $2M_x-1$, which results in the following corollary.
    
\begin{cor}[Non-Sparse FIR Filters]
\label{cor:fir}
Assume the hypotheses of Theorem~\ref{theorem1}. Let $L = M_x$. Then the same identifiability result as in Theorem~\ref{theorem1} holds if and only if $|\mathcal{K}| \geq 2M_x-1$.
\end{cor}
    
Compared with the results of \cite{xu_kilath}, where $3M_x$ time-samples are sufficient to identify the filters, in our frequency-domain approach $2M_x-1$ Fourier measurements are necessary and sufficient. Further, the Fourier measurements can be computed uniquely from $2M_x-1$ time-measurements by using a sampling kernel as shown in Section~\ref{sec:compressive_measurements}. Hence, for a large $M_x$, we gain significantly in terms of reducing the number of measurements compared with the approach in \cite{xu_kilath}. A detailed comparison of these methods is presented in Section~\ref{sec:comparision}.

Next we show how the identifiability result in Theorem~\ref{theorem1} for the two-channel case generalizes to the case of more than two channels. We obtain a particular sufficient condition by assuming that there exists a pair of coprime channels from which at least $2L^2$ Fourier measurements are available. 

\begin{theorem}
\label{theorem2}
Suppose (A1), (A3), (A4), and (A5) hold for $N \geq 2$ and $L<\sqrt{M_x}$. Then compressive MBD is partially identifiable from the Fourier measurements if the following conditions are satisfied: i) There exist $1 \leq n_1 < n_2 \leq N$ such that $\mathcal{K}_{n_1} = \mathcal{K}_{n_2} = \mathcal{K}$ for a universal set $\mathcal{K}$ with $|\mathcal{K}| \geq \min \{2L^2, 2M_x-1\}$; ii) $\mathcal{K}_n$ is a universal set such that $\mathcal{K}_n \subseteq \mathcal{K}$ and $|\mathcal{K}_n| \geq 2L$ for all $n \not\in \{n_1, n_2\}$. 
\end{theorem}
	
\begin{proof}
By Theorem~\ref{theorem1}, the triplet $(\{S(e^{\mathrm{j}k \omega_0})\}_{k\in \mathcal{K}},x_{n_1},x_{n_2})$ is uniquely identified. Then for an appropriate choice of the sets $\mathcal{K}_n$, the recovery of the filters for the rest of the channels (those indexed by $n \not \in \{n_1, n_2\}$) reduces to a non-blind problem. Since $Y_n(e^{\mathrm{j}k \omega_0}) = S(e^{\mathrm{j}k \omega_0}) X_n(e^{\mathrm{j}k \omega_0})$, if the sets $\mathcal{K}_n$ are chosen such that $\mathcal{K}_n \subseteq \mathcal{K}_{n_1}$, then we can compute $\{X(e^{\mathrm{j}k \omega_0})\}_{k\in \mathcal{K}_n}$ from $\{Y_n(e^{\mathrm{j}k \omega_0})\}_{k\in \mathcal{K}_n}$ as $\{S(e^{\mathrm{j}k \omega_0})\}_{k \in \mathcal{K}_{n_1}}$ is already identified. Next, from the measurements $\{X(e^{\mathrm{j}k \omega_0})\}_{k\in \mathcal{K}_n}$, the filters are identifiable uniquely if $|\mathcal{K}_n| \geq 2L$.
\end{proof}

Theorem~\ref{theorem2} implies that a total $4L^2+(N-2)2L$ Fourier measurements from $N$ channels are sufficient for unique identification of the filters. Therefore on average it suffices to take $(4L^2+(N-2)2L)/N$ measurements per channel. Particularly, with sufficiently many channels ($N \gtrsim L$), the average number of measurements per channel is of the order of $L$ which is the same as required by the necessary condition. Note that when the source is known, a minimum of $2L$ Fourier measurements are necessary to uniquely identify the filters. Hence, for $N \gtrsim L$ the requirement on the average number of measurements per channel matches that of the known source case in order. 
	
\subsection{Recovering the Common Source Signal}
\label{ssec:source_recovery}
As discussed in the introduction, in certain applications such as radar, sonar, and ultrasound, it suffices to identify only the filters, which describe the target. On the other hand, there exist applications where the recovery of the source signal is important. For example, in communications or imaging, the source signal carries information and the filters describe the channel impulse response or sensitivity functions. In this section, we present conditions under which the source and filters are simultaneously identified. 
	
The results in Theorems~\ref{theorem1} and \ref{theorem2} guarantee that the filters are fully identified but the source signal is partially identified up to its Fourier measurements at the selected frequencies. In general, the recovery of the source $s$ from its partial Fourier measurements is ill-posed. However, recovery becomes feasible by introducing further restrictions on $s$. 
	
For example, suppose that $s$ is supported within $[M_s]$ for a finite integer $M_s > 0$. To uniquely determine an arbitrary source signal $s$ supported within $[M_s]$, the number of measurements needs to be at least $M_s$. On the other hand, since we choose $\omega_0$ and $\mathcal{K}$ such that the elements in $\{e^{\mathrm{j}\omega_0 k}\}_{k=0}^{M_s-1}$ are distinct, the linear system that generates the Fourier measurements at $\mathcal{K}\omega_0$ corresponds to a Vandermonde matrix of full column rank and $s$ is uniquely determined. 

Combining the above argument with Theorem~\ref{theorem1} provides the following result in the two-channel case: All $s$, $x_1$, and $x_2$ are identified from the sampling pattern given by $\mathcal{K}_1 = \mathcal{K}_2 = \mathcal{K}$ if and only if $|\mathcal{K}| \geq \max\{M_s,2L^2\}$ and (A5) is satisfied.	
	
Below we show that when there are more than two channels, with a carefully designed sampling pattern, one can significantly reduce the peak number of measurements per channel, where the gain is almost proportional to the number of channels. This is interesting, particularly when $M_s$ dominates $L$, that is, $M_s \gg 2L^2$. 
	
A na\"{i}ve approach is to consider $M_s$ Fourier measurements from any pair of channels and to apply Theorem~\ref{theorem1} to identify the corresponding filters and the source $s$. As the source is identified, the problem is reduced to a non-blind one for the remaining channels, and to identify the filters in those channels, $2L$ Fourier measurements are necessary and sufficient. However, this naive approach may not be well suited for practical applications. For example, in the radar and ultrasound applications, the Fourier measurements are computed as follows: The analog signal is first pre-filtered with a kernel followed by an analog-to-digital converter (ADC). Then, as discussed in Section \ref{sec:compressive_measurements}, the time-domain samples are linearly combined to give the Fourier measurements \cite{eldar_sos, mulleti_kernal}. Here the sampling rate is determined by the number of Fourier measurements. In practice, it has been shown that the bit resolution of ADCs is limited when the sampling rate is high \cite{walden_adc}. Therefore, it is desirable to minimize the maximum number of Fourier measurements per channel. 
	
To achieve recovery we consider a pairwise strategy. For example, assume there are four channels ($N = 4$). Consider universal sets $\{\mathcal{K}_n \}_{n=1}^4$ such that they satisfy the following conditions: (i) $\mathcal{K}_1 = \mathcal{K}_2$ and $\mathcal{K}_3 = \mathcal{K}_4$ ; and (ii) $|\mathcal{K}_n| \geq 2L^2$ for $n = 1, 2, 3, 4$. By applying Theorem~\ref{theorem1}, independently to the measurements from the pair channels $(1, 2)$ and $(3, 4)$, we identify the filters $\{x_n\}_{n=1}^4$ as well as the Fourier measurements $\{s(e^{\mathrm{j}k\omega_0})\}_{k \in \mathcal{K}_1}$ and $\{s(e^{\mathrm{j}k\omega_0})\}_{k \in \mathcal{K}_3}$. These two sets of partial Fourier measurements may be differently scaled due to shift and scaling ambiguities. Let us assume that there are no such inter-pair ambiguities. Then we have overall $\{s(e^{\mathrm{j}k\omega_0})\}_{k \in \mathcal{K}_1 \cup \mathcal{K}_3}$ Fourier measurements of the source. If $|\mathcal{K}_1 \cup \mathcal{K}_3| \geq M_s$, then we can uniquely recover the source. Here, the maximum number of the Fourier measurements can be $2L^2$. In this particular case, if we consider more channels, it is necessary and sufficient to consider $2L$ measurements from the additional channels to identify the corresponding filters as the source is identified from the first four channels.  We generalize the example to any $N$ channels and show how to choose the universal sets $\{\mathcal{K}_n\}_{n=1}^N$ to eliminate the inter-pair ambiguity and uniquely identify the source and the filters.
	
To this end, to apply the pairwise strategy for any $N\geq 2$ channels, we consider the Fourier-domain sampling grids given by
\begin{equation}
\label{eq:calK_n}
\begin{aligned}
\mathcal{K}_{2r-1} = \mathcal{K}_{2r},|\mathcal{K}_{2r-1}|= K, & ~\mathcal{K}_{2r-1} ~ \text{is a universal set,} \\
\text{and}\quad  |\mathcal{K}_{2r-1} \cap \mathcal{K}_{2r+1}|&=2, \quad r = 1, 2, \dots, R,
\end{aligned}
\end{equation}
where $R \leq \lfloor N/2 \rfloor$. The first three conditions with $K \geq 2L^2$ are necessary and sufficient for the recovery of the filters and the partial Fourier measurements of the source. The last condition is that there should be overlap of two samples in successive pairs of sample sets. We show that this overlap aids in removing inter-pair ambiguity of shift and scaling. 

The identifiability result for both the source and filter for compressive MBD is stated in the following theorem.

\begin{theorem}
\label{thm_source}
{Suppose that (A1) to (A5) hold for $N\geq 2$ and $L<\sqrt{M_x}$. Then compressive MBD is uniquely identifiable according to Definition~\ref{def:identifiability} from the Fourier measurements if the following conditions are satisfied: i) For at least $\max \left \{\frac{M_s-K}{K-2}+1, 1 \right \}$ pair of channels, the corresponding sampling sets satisfy the conditions in \eqref{eq:calK_n} with $K \geq 2L^2$; and ii) At least $2L$ Fourier measurements are available from the rest of the channels.}
\end{theorem}

\begin{proof}
Let us assume that we can recover $M_s$ Fourier measurements of $s$ from the first $2R$ channels together with the corresponding filters where $2 R \leq N$. For these channels, the sampling pattern is chosen such that $\mathcal{K}_{2r-1} = \mathcal{K}_{2r}$ for $r = 1,\dots,R$ with $|\mathcal{K}_{2r}| = K \geq 2L^2$. For each pair of $(2r-1)$th and $2r$th channels for $r = 1,\dots, R$, the assumptions imply via Theorem~\ref{theorem1} that $x_{2r-1}$, $x_{2r}$, and $S(e^{\mathrm{j}k\omega_0})$ for $k \in \mathcal{K}_{2r-1}$ are uniquely identified up to a scaling and shift ambiguity. In other words, $S(e^{\mathrm{j}k\omega_0})$ is identified up to multiplication by $\alpha_r e^{\mathrm{j}k\omega_0 p_r}$ for $k \in \mathcal{K}_{2r-1}$ for unknown constants $\alpha_r \neq 0$ and $p_r \in \mathbb{Z}$. 
Due to the overlaps $|\mathcal{K}_{2r-1} \cap \mathcal{K}_{2r+1}|=2$ for $r=1, \dots, R$ in the design of $\{\mathcal{K}_{2r-1}\}_{r=1}^{R}$. These inter-pair ambiguity constants can be removed up a global constant in a sequential manner. For example, let us assume that $\mathcal{K}_1 \cap \mathcal{K}_3 = \{k_1, k_2\}$. In other words, for the channel pairs $(1, 2)$ and $(3, 4)$, Fourier measurements are taken at the overlapped frequencies $k_1\omega_0$ and $k_2\omega_0$. By applying Theorem~\ref{theorem1} to these pairs, we obtain the Fourier measurements of the source at the overlapped frequencies up to inter-pair ambiguities, which are $\alpha_1 S(e^{\mathrm{j}k_2\omega_0}) e^{\mathrm{j}k_2\omega_0 p_1}$, $\alpha_1 S(e^{\mathrm{j}k_1\omega_0}) e^{\mathrm{j}k_1\omega_0 p_1}$, $\alpha_2 S(e^{\mathrm{j}k_2\omega_0}) e^{\mathrm{j}k_2\omega_0 p_2}$, and $\alpha_2 S(e^{\mathrm{j}k_1\omega_0}) e^{\mathrm{j}k_1\omega_0 p_2}$. Then $\alpha_1/\alpha_2$ and $p_1-p_2$ are computed from the ratios among these measurements and enable to obtain $\{\alpha_1 S(e^{\mathrm{j}k\omega_0})e^{\mathrm{j}kp_1\omega_0} \}_{k \in \mathcal{K}_3}$ from $\{\alpha_2 S(e^{\mathrm{j}k\omega_0})e^{\mathrm{j}kp_2\omega_0} \}_{k \in \mathcal{K}_3}$, where the former is aligned to the first pair.
		
Applying this process successively, we identify the filters $x_1, x_2, \dots, x_{2R}$, up to a global scaling factor $1/\alpha_1$ and a shift by $-p_1$, together with the Fourier measurements $\alpha_1 S(e^{\mathrm{j}k\omega_0})e^{\mathrm{j}kp_1\omega_0}$ for $k \in \bigcup_{n=1}^{R} \mathcal{K}_n$. 
		
To identify $s$ from the above Fourier measurements, it is sufficient to satisfy 
\begin{align}
\bigcup_{r=1}^{N_1} \mathcal{K}_{2r-1}=  (R - 1)  (K-2)+K \geq M_s.
\label{eq:source_channels}
\end{align}
In other words the source can be identified from a minimum $2R \geq 2 \cdot \frac{M_s-K}{K-2}+2$ channels if $K \geq 2L^2$. From the remaining $N-2R$ channels it is sufficient to consider any $2L$ Fourier measurements to identify the corresponding filters. \end{proof}

The maximum number of measurements per channel in Theorem~\ref{thm_source} can be restricted to $2L^2$. The inequality in \eqref{eq:source_channels} implies that when $K$ increases beyond $2L^2$, the number of channels for the identifiability can be reduced. In other words, one can trade-off between the number of measurements per channel and the number of channels. For $K=2L^2$ the source and the filters are identifiable if $N \geq \max\left\{\frac{M_s - 2L^2}{L^2-1} +2, 2 \right\} $.   
	
By Corollary~\ref{cor:fir}, we obtain an immediate extension of Theorem~\ref{thm_source} to the non-sparse case.

\begin{cor}[Identifiability Results for Source and Non-Sparse Filters]
\label{cor:source_nonsparse}
In Theorem~\ref{thm_source}, let $\text{supp}\{x_n\} = [M_x]$, where $M_x\leq M/2$. Then compressive MBD is uniquely identifiable from $2M_x-1$ Fourier measurements from each of $N\geq 2 \cdot \frac{M_s-2M_x+1}{2M_x-3}+2$ channels. 
\end{cor}

Even in the non-sparse case, the measurement system can be compressive when the number of Fourier measurements $2M_x-1$ is smaller than the available time-domain measurements $M= M_x+M_s-1$, that is, $M_s>M_x$. 

\subsection{Extension to Sparse MBD with Circular Convolution}
The MBD problem considered in the previous sections assumes that the measurements consist of a linear convolution of the source and filters, whereas, the recent results in the literature consider the MBD problem with circular convolutions. Here we extend our results to the case of circular convolution. In this setup, the $N$-channel MBD time-domain outputs are given as $\{y_n = s\circledast x_n\}_{n=1}^N$, where $\circledast$ denotes circular convolution, and the supports of the filters and source are within the set $[M]$. In other words, we assume that $\text{supp}\{s\} \subseteq M$ and  $\text{supp}\{x_n\} \subseteq M$. Due to circular convolution, the measurements $y_n = s \circledast x_n$ are $M$-periodic. We further assume that the filters are $L$-sparse and they are coprime. In this case, the goal is to derive identifiability conditions to uniquely recover the source $s$ and the filters $\{x_n\}_{n=1}^N$ from the discrete Fourier transform (DFT) measurements $\{Y_n(e^{\mathrm{j}k\omega_0})\}_{k \in \mathcal{K}_n}$, where $\omega_0 = \frac{2\pi}{M}$ and $\mathcal{K}_n \subseteq [M]$. With these settings, for $N=2$, following the steps in the proof of Lemma~\ref{lemma1}, the sequence $q$ in  \eqref{eq:q} is given as $q = x_1 \circledast \hat{x}_2 - x_2 \circledast \hat{x}_1$. To follow the remaining steps of the proof and prove the identifiability results, we have to ensure that $q = x_1 \circledast \hat{x}_2 - x_2 \circledast \hat{x}_1=x_1 * \hat{x}_2 - x_2 \ast \hat{x}_1$. This is indeed true if we assume that the filters are supported within the set $\lfloor M/2 \rfloor$. With these assumptions, we state the extension of Theorem~2 for the circular convolution case.

\begin{theorem}
\label{theorem04}
{Let $N\geq 2$ and $M\in \mathbb{N}$. Let $(s, \{x_n\}_{n=1}^N)$ be arbitrary while satisfying (A1) to (A5) with $M_x = \lfloor M/2 \rfloor$ and $M_s = M$. Let $y_n = s\circledast x_n$ for $n=1, 2, \dots, N$. Then $s$ and $\{x_n\}_{n=1}^N$ are simultaneously identified from the DFT measurements $\{Y_n(e^{\mathrm{j}\omega_0k})\}_{k \in \mathcal{K}_n}$ if the following conditions are satisfied: (i) For at least $\max \left \{\frac{M-K}{K-2}+1, 1 \right \}$ pair of channels, the corresponding sampling sets $\mathcal{K}_n$ satisfy the conditions in \eqref{eq:calK_n} where $K \geq 2L^2$; and (ii) a minimum of $2L$ DFT measurements are available from the rest of the channels.}
\end{theorem}

Then the following result for the non-sparse case is obtained as an immediate corollary. 

\begin{cor}[Non-Sparse FIR with Circular Convolution]
\label{cor:fir_circle}
Consider the assumptions of Theorem~\ref{theorem04}. Let $\text{supp}\{x_n\} = [M_x]$ where $M_x \leq M/2$. Suppose that $N\geq \frac{2M-4}{2M_x-3}$. Then both the source and the filters are identifiable iff the number of Fourier measurements are greater then or equal to $2M_x-1$.
\end{cor}
Note that the condition $\text{supp}\{x_n\} = [M_x]$ implies that the filters are in a low-dimensional subspace of dimension $M_x$.
	
\subsection{Recovery from Samples in the $z$-Domain}
A DTFT can be considered as a special case of the $z$-domain sample evaluated at a complex number of unit modulus. In this section we show that the results in the previous sections generalize to the case where the measurements of the output channels are given as samples in the $z$-domain

For example, as in \eqref{eq:optNEW}, let us consider the filter identification problem for the two-channel case from samples $\{Y_n(z_k) = S(z_k) X_n(z_k)\}_{k\in [K]}$ where $\{z_k\}_{k\in[K]}$ denotes sampling grid in the $z$-domain. We assume that $\{S(z_k)\}_{k \in [K]}$ is non vanishing. To show the identifiability, we can follow the lines of the proof of Theorem~\ref{theorem1} by substituting $e^{\mathrm{j}k\omega_0}$ by $z_k$. With the $z$-domain measurements, all the steps of the proof of Theorem~\ref{theorem1} remain valid except the spark properties of the resulting $\mathbf{A}$ matrix (see \eqref{eq:Aq}). 

With $z$-domain sampling, the matrix $\mathbf{A}$ is given as
\begin{align}
\mathbf{A}=\begin{pmatrix}
1 &z_1 & z_1^2 &\hdots & z_1^{2M_x-1}\\
1 & z_2 & z_2^2 &\hdots & z_2^{2M_x-1}\\
\vdots & \vdots & \vdots & \ddots & \vdots\\
1 & z_K & z_K^2 &\hdots & z_K^{2M_x-1}\\
\end{pmatrix} \in \mathbb{C}^{K \times 2M_x }.
\label{eq:amat}
\end{align}
The matrix $\mathbf{A}$ in \eqref{eq:amat} need not have full spark for any arbitrary choice of $z_k$. Here we show a particular choice of $z_k$s such that the matrix has full spark. 
	
Let us assume that $z_0 \in \mathbb{C}$ such that $z_0 \neq 1$. Let $\mathcal{K} \in [2M_x]$ be a distinct set of integers such that $|\mathcal{K}| = K$. Let $z_k = z_0^{p_k}$ for $p_k \in \mathcal{K}$. Then the matrix $\mathbf{A}$ in \eqref{eq:amat} has full spark if $\mathcal{K}$ is a universal set. 
To show the full spark property, let us consider a submatrix of $\mathbf{A}$ which consists of $K$ distinct columns indexed by $m_1, m_2, \dots, m_{K}$. With $z_k = z_0^{p_k}$, the submatrix is given as
\begin{align}
\begin{pmatrix}
z_0^{p_1 m_1} &z_0^{p_1 m_2}  &\hdots & z_0^{p_1 m_K} \\
z_0^{p_2 m_1} & z_0^{p_1 m_2}  &\hdots & z_0^{p_1 m_K} \\
\vdots & \vdots & \ddots & \vdots\\
z_0^{p_K m_1} & z_0^{p_K m_2}  &\hdots & z_0^{p_K m_K} \\
\end{pmatrix} \in \mathbb{C}^{K \times K }.
\label{eq:amat2}
\end{align}
Since the choice of the columns is arbitrary, the matrix $\mathbf{A}$ will have full spark if the submatrix is invertible. Since $z_0 \neq 1$, the submatrix is similar to matrix $\mathbf{V}$ in Section \ref{sec:universal_set} with seeds $\{z_0^{m_k}\}_{k =1}^K$. Since $\mathcal{K}$ is a universal set and $p_k \in \mathcal{K}$, the submatrix has full spark and is invertible. Hence, the matrix $\mathbf{A}$ with the particular choice of sampling grid in the $z$-domain has full spark which implies that the filters can be uniquely identifiable in the two-channel case if $K \geq 2L^2$. The problem is not identifiable if $K < 2L$. Similarly, we can extend the results of Theorem~\ref{theorem2} and Theorem~\ref{thm_source} to the case where the samples are measured in the $z$-domain.

A major difference between the identifiability results from the Fourier measurements and from the measurements in the $z$-domain is that the sampling grid in the former case depends on the support of the source and the filters. For example, the results in Theorem~\ref{thm_source} assumes that the sampling interval $\omega_0$ is selected such that the set $\{e^{\mathrm{j}\omega_0 k}\}_{k=0}^{\max\{2M_x-2,M_s-1\}}$ has distinct elements. However, with $z$-domain sampling, the sampling grid can be designed independent of the support of the source or the filters. 
	
\section{Comparison With Prior Art}
\label{sec:comparision}
We first compare our results for non-sparse cases and then provide a comparison for sparse MBD.

\subsection{Comparison of Non-Sparse Case}
Xu et al. \cite{xu_kilath} considered the problem of estimating filters only without sparsity assumption with linear-convolution. They made the assumptions that the filters are coprime and the source has a linear complexity\footnote{Mathematically, the linear complexity of a sequence $s$ is defined as the smallest integer $L_c$ such that there exists a set of complex-valued amplitudes $\{c_\ell\}_{\ell=1}^{L_c}$ and complex-valued roots $\{r_\ell\}_{\ell=1}^{L_c}$ such that $ s[m] = \sum_{\ell=1}^L c_\ell r_\ell^m$.}  greater than or equal to $2M_x$, which implies that $M_s \geq 4M_x$. With these assumptions, the authors show that the filters are identifiable from $N=2$ channels if $3M_x$ consecutive measurements of $y_n$ are available.  Note that with $M_s \geq 4M_x$, the length of $y_n$ is given by $M \geq 5M_x-1$ out of which $3M_x$ are sufficient to identify the filters.
	
In comparison, our results in Corollary \ref{cor:fir}, together with the time-domain results in Corollary \ref{cor:time}, state that we can identify the filters with $2M_x-1$ time samples. As in \cite{xu_kilath}, we too impose the coprimeness condition on the filters. However, we do not restrict the filter to have a longer support. In our approach,  the support of the source could be either larger or smaller compared with the support of the filters. Furthermore, our results are valid for any source signal whose DTFT samples do not vanish at a given frequency location. The source need not satisfy a linear complexity constraint. 

For example, let us assume that the source $s$ has linear-complexity of one, that is, the samples of the source sequence are given as $s[m] = c_1 r_1^m$ for $m \in [M_s]$. In addition, let us assume that $r_1 = e^{\mathrm{j}\omega_1}$. In this case, the DTFT of the source sequence is given as $S(e^{\mathrm{j}\omega}) = c_1 \frac{1-e^{\mathrm{j}(\omega-\omega_1)M_s}}{1-e^{\mathrm{j}(\omega-\omega_1)}}$. The DTFT vanishes at $\omega = p\frac{ 2\pi}{Ms}+\omega_1$ where $p \in \mathbb{Z}\setminus\{0\}$. If we chose our sampling set $\mathcal{K}$ and $\omega_0$ such that the set $\mathcal{K}\omega_0$ does not have zeros of $S(e^{\mathrm{j}\omega})$ then our method identifies the filters. In particular, let the cardinality of the set $\mathcal{K}$ be given, that is, the number of Fourier measurements to be taken is known. For a given $\omega_1$ and $M_s$, if $\omega_0$ and $\mathcal{K}$ are chosen as $\min \left\{\frac{4\pi}{M_s |\mathcal{K}|}, \frac{2\pi}{2M_x-1} \right\}$ and $ \left \lceil \frac{\omega_1-2\pi/M_s}{\omega_0} \right\rceil +[|\mathcal{K}|]$, then $\{S(e^{\mathrm{j}k\omega_0})\}_{k\in \mathcal{K}}$ does not vanish. Hence, our approach identifies the filters for $L_c=1$, whereas, the method by Xu et al. \cite{xu_kilath} cannot identify the filters uniquely.
	
Recently, Xia and Li \cite{xia_li} considered MBD problem with circular convolution in a deterministic setup where the source and filters are real and deterministic, but not sparse. Similar to our assumption in the case of non-sparse circular MBD setup, they assumed that $\text{supp}\{x_n\} = [M_x]$ with $M_x<M$. They showed that by using all $M$ time samples from each channel \emph{almost all} the sources and the filters are uniquely identifiable iff  $N\geq \frac{M-1}{M-M_x}$. The authors used the conjugate symmetry property of the Fourier transform of real signals to restrict the feasible solution sets. The results show that there exist source and filters such that the identifiability results fail. However, our results hold for any source filter pairs, which satisfy the desired conditions (cf. Corollary~\ref{cor:fir_circle}). We show that $N\geq \frac{2M-4}{2M_x-3}$ channels are sufficient for unique identification of the source and filters iff $2M_x-1$ Fourier measurements per channel are available. For $M_x = M/2$, both the results work for $N=2$ channels and require all the measurements. On the other hand, when $M_x \ll M$, our result requires more number of channels but with fewer measurements per channel than that by Xia and Li \cite{xia_li} with two channels.

\subsection{Comparison with Recent Results on Sparse MBD}
All the results discussed in this section consider the MBD setup with circular convolution. Hence, we compare them with our results in Theorem~4. The results in \cite{balzano_nowak, lee_17, wang_chi, cosse_mbd} consider identifiability of MBD problems with the assumption that the filters are random and sparse. Balzano and Nowak \cite{balzano_nowak} considered a BGPC problem with oversampled DFT matrix and showed that when $L$-sparse signals $x_n$s are generic and have common known support, it is necessary to have measurements from $N\geq \lceil (M-1)/(M-L) \rceil $ channels for perfect recovery of unknown gains. Li et al. \cite{lee_17} studied the identifiability conditions for a general BGPC problem with subspace and sparsity constraints. The authors showed that under generic sparsity constraints on the filters, the problem is identifiable with high-probability up to acceptable ambiguities as long as $N=\mathcal{O}(M \log M)$. In \cite{wang_chi, cosse_mbd}, the authors considered a sparse MBD problem with circular convolution by assuming that the source is invertible and the filters are sparse with randomly chosen support and amplitudes. Specifically, Wang and Chi \cite{wang_chi} assumed that the sparsity of the filters follows a Bernoulli-subgaussian model. With the assumption that the source $s$ is approximately flat in the Fourier domain, the authors show that the problem could be efficiently solved through an $\ell_1$-minimization approach, as long as $N=\mathcal{O}(M \log^4 M)$. Cosse \cite{cosse_mbd} assumed that the source is invertible and the location of the non-zero values of the sparse filters are chosen uniformly at random over $[M]$. The recovery is guaranteed with high-probability as soon as the number of channels $N$ and the dimension of the filters $M$, satisfy $N \lesssim M$ and $N \gtrsim L^2$ where the filters are assumed to be $L$-sparse.\footnote{$M \gtrsim N \Leftrightarrow \exists c \in \mathbb{R}$ s.t. $M \geq c N$.} 

Our results is distinguished from the previous results as follows:
\begin{enumerate}
	\item The aforementioned recent results considered random sparsity or subspace models on the filters, whereas, we consider a deterministic sparsity model for the filters.
	
	\item In \cite{balzano_nowak,lee_17, wang_chi, cosse_mbd}, sparsity is introduced to derive the identifiability results but not with the goal of compressing the measurements. The results are derived by assuming that all $M$ time samples $y_n$ are available in all the channels. We show that sparsity also helps in identifiability by using compressive measurements in the frequency domain. Instead of using $M$ time samples, we show identifiability by using $2L^2$ Fourier samples of $y_n$. 
	
	\item In \cite{lee_17} and  \cite{wang_chi}, the number of channels required does not depend on the sparsity level of the filters. We show that the source and filters are identifiable from $N \geq \max \left\{\frac{M - 2L^2}{L^2-1} +2, 2 \right\}$ channels. 
		
	\item In \cite{lee_17} and  \cite{wang_chi}, the total number of measurements is on the order of $\mathcal{O}(M \log M)$ and $\mathcal{O}(M \log^4 M)$, respectively, whereas, in our setup, we need $2L^2$ measurements from at least  $\max \left\{\frac{M - 2L^2}{L^2-1} +2, 2 \right\}$ channels, which results overall $\frac{(M-2)2L^2}{L^2-1}$ measurements which is on the order of $M$. 
		
	\item Comparing the results in the case of known support, in \cite{balzano_nowak}, all $M$ samples are considered per channel and overall $\mathcal{O}(M)$ measurements are required. In our case, we need only $L^2$ measurements per channel with overall measurements on the order $M$. We gain in terms of the number of measurements per channel but we require more channels compared with \cite{balzano_nowak}.

\end{enumerate}

\subsection{Relation to Blind CS}
The proposed compressive MBD problem can be viewed as a special case of blind compressive sensing (BCS) \cite{blindCS}. In BCS, a set of signals that are sparse in an unknown bases are uniquely identified from their compressed measurements. Similarly, in compressive MBD, the output sequences $\{y_n = s*x_n\}_{n=1}^N$ are sparse in the unknown dictionary that is given by the convolution matrix corresponding to the sequence $s$, the filters denote the sparse vectors, and the objective is to identify them from compressive measurements in the Fourier domain. However, the existing dictionary models in \cite{blindCS} do not include convolutional dictionaries.

\section{Numerical Results}
\label{sec:simulations}
The main goal of this section is to compare the identifiability results of Theorem~\ref{theorem1} to empirical observations.
While Theorem~\ref{theorem1} provides a set of necessary and sufficient conditions for the identifiability in the worst case, in practice, it is not feasible to test in the worst case scenario. Therefore, instead, we run a set of Monte Carlo simulations and observed conditions under which most cases are successful (We counted the frequency of empirical successes). On the other hand, the identifiability result by Theorem~\ref{theorem1} implies the existence of a method that uniquely determines the solution of compressive MBD, which means one has to consider the optimal algorithm regardless of its computational cost. This is also infeasible in practice. For a set of small sized problems, we performed enumeration over all possible supports, which provides an optimal reconstruction algorithm. For larger scaled problems, we consider a few selected heuristics described below. In the non-sparse case with noisefree measurements, an optimal algorithm can be obtained by a standard eigenvalue decomposition. Thus, we use it for the study of the identifiability. 

\subsection{Practical algorithms for compressive MBD}
To describe the algorithms used in the experiments, we rewrite the measurements succinctly in a compact matrix form. 
Let $\mathbf{y}_n \in \mathbb{C}^{|\mathcal{K}|}$ denote the compressible measurements from the $n$th channel. Let $\mathbf{s} \in \mathbb{C}^{|\mathcal{K}|}$ denote the column vector whose entries are $\{S(e^{\mathrm{j}k\omega_0})\}_{k \in \mathcal{K}}$. The measurements in \eqref{eq:new_mes} can then be compactly written as
\begin{align}
    \centering
    \mathbf{Y} = [\mathbf{y}_1\,\,  \mathbf{y}_2] = \text{diag}(\mathbf{s})\mathbf{\bar{A}} \mathbf{X},
    \label{eq:mes_mat}
\end{align}
where $\mathbf{X} = [\mathbf{x}_1 \,\, \mathbf{x}_2]$. Estimating $\mathbf{s}$ and $\mathbf{X}$ from $\mathbf{Y}$ is a BGPC or BDC problem with partial Fourier matrix. Practical algorithms to solve the sparse MBD or BGPC/BDC problem were proposed in \cite{wang_chi, gribonval_dl, li2017blind}. 
The algorithm proposed in \cite{wang_chi} is based on $\ell_1$-norm minimization and requires the matrix $\bar{\mathbf{A}}$ to be a full DFT matrix. Hence, the algorithm is not applicable to compressive MBD. In \cite{gribonval_dl} and \cite{li2017blind}, the authors proposed algorithms based on BDC and truncated-power iteration (TPI), respectively. 
We will use modifications of these two approaches to our setting in the simulation. 

The BDC problem is similar to a sparse dictionary learning (DL) problem \cite{ksvd}; however, there is a major difference. In the DL problem, the composite matrix $\text{diag}(\mathbf{s})\bar{\mathbf{A}}$ is unknown and needs to be estimated along with the sparse vectors $\mathbf{X}$. In contrast, BDC requires to estimate only $\mathbf{s}$ with known $\bar{\mathbf{A}}$. In \cite{gribonval_dl}, the authors proposed a convex optimization based solution to the BDC problem. However, we experimentally observed that for $N=2$ solving the BDC problem in \eqref{eq:mes_mat} by using an alternate minimization approach provides desirable identifiability results. Hence, we consider the alternate minimization-based BDC approach as one of the practical algorithms for assessing the results. 
In alternating minimization, first, we estimate the sparse vectors by assuming that $\mathbf{s}$ is known and then, in the dictionary update step, we estimate $\mathbf{s}$ by using the estimated $\mathbf{X}$. We apply orthogonal matching pursuit (OMP) \cite{pati1993orthogonal} to estimate the sparse vectors. Then $\mathbf{s}$ is estimated as a minimizer of the error $\|\mathbf{Y}-\text{diag}(\mathbf{s})\bar{\mathbf{A}}\mathbf{X}\|_2^2$. We present the proposed approach in Algorithm~\ref{alg:dl}. 
\begin{algorithm}[!h]
\begin{algorithmic}[1]
\caption{BDC for solving \eqref{eq:mes_mat}.}
\label{alg:dl}
\Statex {\textbf{Output:} $\mathbf{s}$ and $\mathbf{X}$}
\Statex {\textbf{Input:} $\mathbf{Y}$, $\bar{\mathbf{A}}$ $L$, and the initial estimate $\mathbf{s}^{(0)}$}
\State Let $i \leftarrow 1$
\Repeat
\State Estimate $\mathbf{X}^{(i)}$ by applying OMP to $\text{diag}(\mathbf{s}^{(i-1)})^{-1}\,\mathbf{Y}$ columnwise
\State $\mathbf{s}^{(i)} \leftarrow \mathop{\rm{argmin}}_{\mathbf{s}} \| \mathbf{Y}- \text{diag}(\mathbf{s})\bar{\mathbf{A}}\mathbf{X}^{(i)} \|_2^2$
\State $i \leftarrow i+1$
\Until convergence criterion is reached
\end{algorithmic}
\end{algorithm}
The solution of the optimization problem in Step 4 is given as $\mathbf{s}^{(i)} = \mathcal{D}\{\bar{\mathbf{A}}\mathbf{X}^{(i)}\mathbf{Y}^{\mathrm{H}}\}./ \mathcal{D}\{(\bar{\mathbf{A}}\mathbf{X}^{(i)})(\bar{\mathbf{A}}\mathbf{X}^{(i)})^{\mathrm{H}}$\}. Here the operator $\mathcal{D}$ acts on a square matrix to output a vector consisting of the diagonal elements of the matrix, and the symbol $./$ denotes element-wise division. In our simulations, the algorithm stops at the $i$th iteration if $\|\mathbf{X}^{(i)}-\mathbf{X}^{(i-1)}\|_2 \leq 10^{-3}$.

Next, we discuss an alternative method to identify the filters from the measurements $\mathbf{Y}$. By applying cross-correlation, as in the proof of our main results (cf. Section~\ref{sec:proof}), \eqref{eq:mes_mat} can be rewritten as
\begin{align}
\underbrace{\begin{bmatrix} \text{diag}(\mathbf{y}_2) \mathbf{\bar{A}} & -\text{diag}(\mathbf{y}_1) \mathbf{\bar{A}} \end{bmatrix}}_{\mathbf{B}} \, \underbrace{\begin{bmatrix} \mathbf{x}_1 \\ \mathbf{x}_2 \end{bmatrix}}_{\boldsymbol{\gamma}} = \mathbf{0}.
\label{eq:cc_mat}
\end{align}
In the sparse MBD framework, identifying the filters from the matrix $\mathbf{B}$ is equivalent to identifying a $2L$-sparse null vector of $\mathbf{B}$. The solution to the problem in \eqref{eq:cc_mat} can be computed as the solution to the following non-convex optimization problem:
\begin{align}
\hspace{-.1in}\begin{array}{ll}
\displaystyle \mathop{\mathrm{minimize}}_{\boldsymbol{\gamma}_1, \boldsymbol{\gamma}_2 \in \mathbb{C}^{M_x}} &  \begin{bmatrix} \boldsymbol{\gamma}_1^{\mathrm{H}} & \boldsymbol{\gamma}_2^{\mathrm{H}} \end{bmatrix} \mathbf{B}^{\mathrm{H}}\mathbf{B} \begin{bmatrix} \boldsymbol{\gamma}_1 \\ \boldsymbol{\gamma}_2 \end{bmatrix} \\
\mathrm{subject~to} & \| \boldsymbol{\gamma}_1\|_0 \leq L,\,\,  \| \boldsymbol{\gamma}_2\|_0 \leq L, \,\, \left\|\begin{bmatrix} \boldsymbol{\gamma}_1 \\ \boldsymbol{\gamma}_2 \end{bmatrix}\right\|_2 = 1.
\end{array}
\label{eq:tpi_opt}
\end{align}
Problem \eqref{eq:tpi_opt} can be solved by adapting the truncated power iteration (TPI) algorithm proposed in \cite{yuan_tpi}. In its original version, TPI is developed to compute the largest sparse eigenvector of a positive semidefinite matrix. In \cite{li2017blind}, Li et al. adopt the TPI algorithm in \cite{yuan_tpi} to solve the BGPC problem by assuming that measurement matrix is random Gaussian and the filters are jointly sparse. We present TPI to solve the optimization problem in \eqref{eq:tpi_opt}.

The TPI algorithm in \cite{yuan_tpi} is developed for a single sparse vector. In \eqref{eq:tpi_opt} the vector to be estimated is a concatenation of two sparse vectors. By adopting the original TPI algorithm, the sparsity constraints on $\boldsymbol{\gamma}$ in \eqref{eq:tpi_opt} are imposed by the composite sparse projector $\mathscr{\widetilde{P}}_L: \mathbb{C}^{2M_x}\rightarrow \mathbb{C}^{2M_x}$. For any vector $\boldsymbol{\gamma}\in \mathbb{C}^{2M_x}$, the output of the sparse projector, $\mathscr{\widetilde{P}}_L\{\boldsymbol{\gamma}\}$, is an $2L$ sparse vector computed by independently retaining the $L$ largest entries over the sets $[M_x]$ and $M_x+[M_x]$ from the support of $\boldsymbol{\gamma}$ and setting the rest of the entries to zero. The truncated power iteration method is presented in Algorithm~2.

\begin{algorithm}[h]
\begin{algorithmic}[1]
\label{alg:algorithm}
\caption{Truncated Power Iteration to solve \eqref{eq:tpi_opt}.}
\Statex {\textbf{Output:} $\boldsymbol{\gamma}$}; {\textbf{Parameter:} $\beta$}
\Statex {\textbf{Input:} $\mathbf{B}, M_x, L$, and the initial estimate $\boldsymbol{\gamma}^{(0)}$} 
\State Set $\mathbf{G} \leftarrow \beta \mathbf{I}_{2M_x} -  \mathbf{B}^{\mathrm{H}}\mathbf{B}$
\State Let $i \leftarrow 1$
\Repeat
\State $\boldsymbol{\gamma}^{(i)} \leftarrow \mathbf{G}\boldsymbol{\gamma}^{(i-1)}/\|\mathbf{G}\boldsymbol{\gamma}^{(i-1)}\|_2$
\State $\boldsymbol{\gamma}^{(i)} \leftarrow \mathscr{\widetilde{P}}_L\{\boldsymbol{\gamma}^{(i)}\}/\| \mathscr{\widetilde{P}}_L\{\boldsymbol{\gamma}^{(i)}\} \|_2$
\State $i \leftarrow i+1$
\Until convergence criterion is reached
\end{algorithmic}
\end{algorithm}
In Step 2, we denote by $\mathbf{I}_{2M_x}$ the identity matrix of size $2M_x \times 2M_x$. As was suggested in \cite{li2017blind}, a safe choice of the parameter $\beta$ is $\|\mathbf{B}\|$. We initialize $\boldsymbol{\gamma}$ as the concatenation of the outputs of the OMP algorithm to the inputs $\mathbf{y}_1$ and $\mathbf{y}_2$. The algorithm stops when the update in $\boldsymbol{\gamma}$ is not significant in successive iterations. Specifically, we stop at the $i$th iteration if $\|\boldsymbol{\gamma}^{(i)}-\boldsymbol{\gamma}^{(i-1)}\|_2 \leq 10^{-3}$.

\begin{figure}[!t]
\centering
\includegraphics[width = 1.6in]{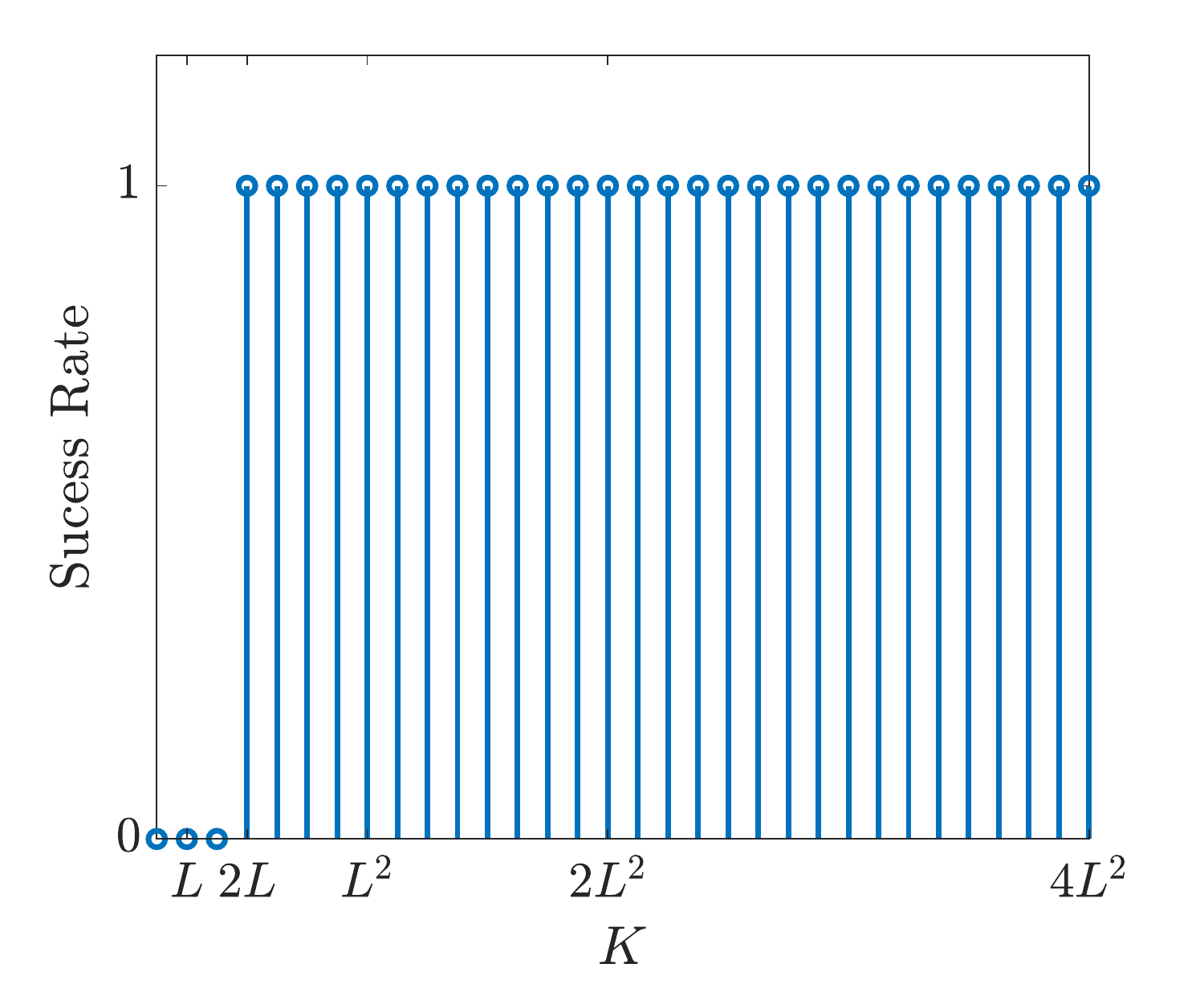}
\caption{Phase transition for the identification of $L$-sparse filters by exhaustive search ($L=4$): $K \geq 2L$ is necessary and sufficient.}
\label{fig:Esearch}
\end{figure}

\begin{figure}[!h]
\centering
\includegraphics[width = 2in]{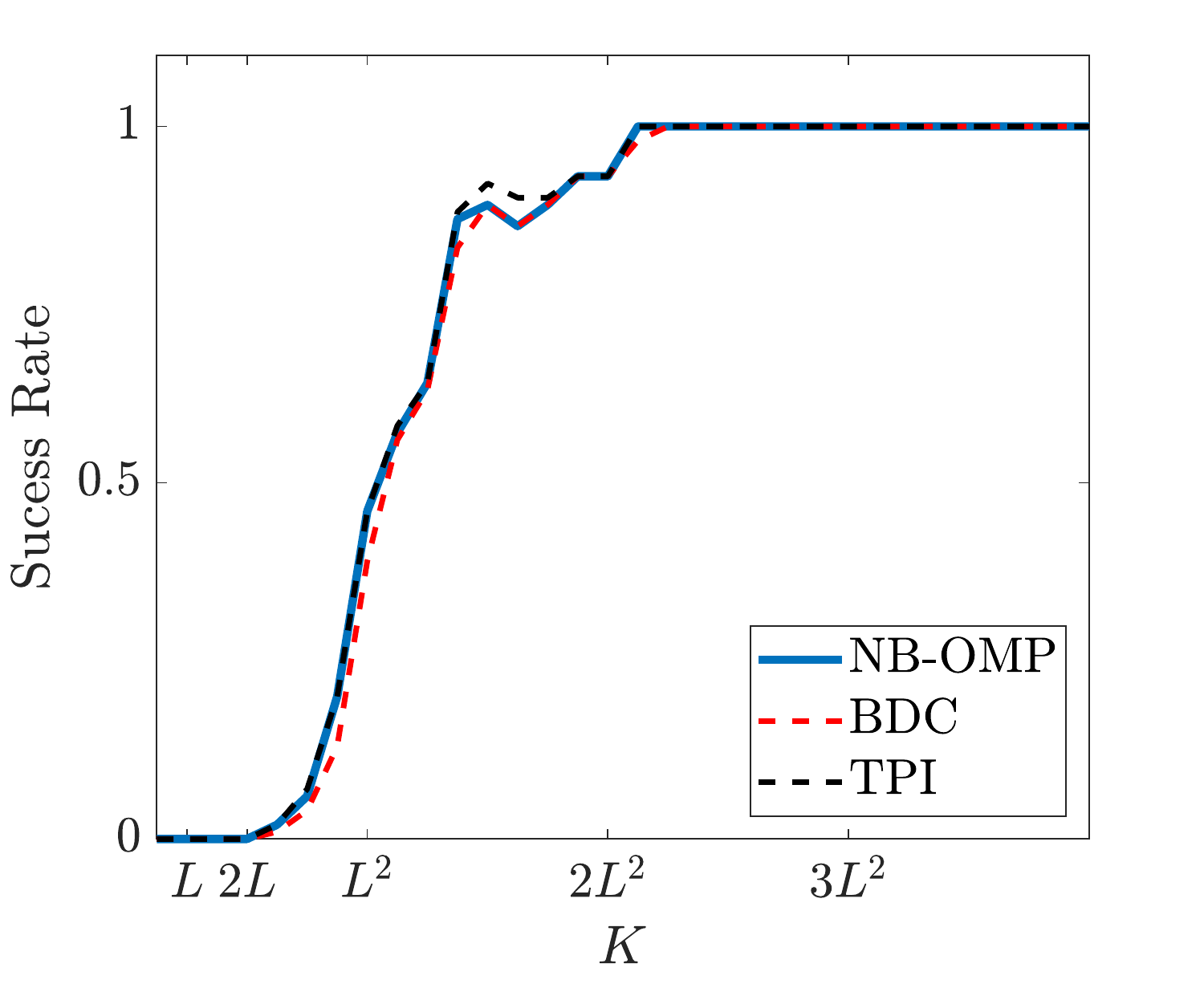}
\caption{Phase transition for the identification of $L$-sparse filters by NB-OMP, BDC, and TPI with known source ($L=4$, $M_x = 2L^2$); $K \geq 2L$ is necessary and $K>2L^2$ is sufficient.}
\label{fig:L4_comp}
\end{figure}

\subsection{Comparison of Sparse MBD}
\begin{figure}[!t]
\begin{center}
\begin{tabular}{ccc}
\subfigure[NB-OMP]{\includegraphics[width = 1.6in]{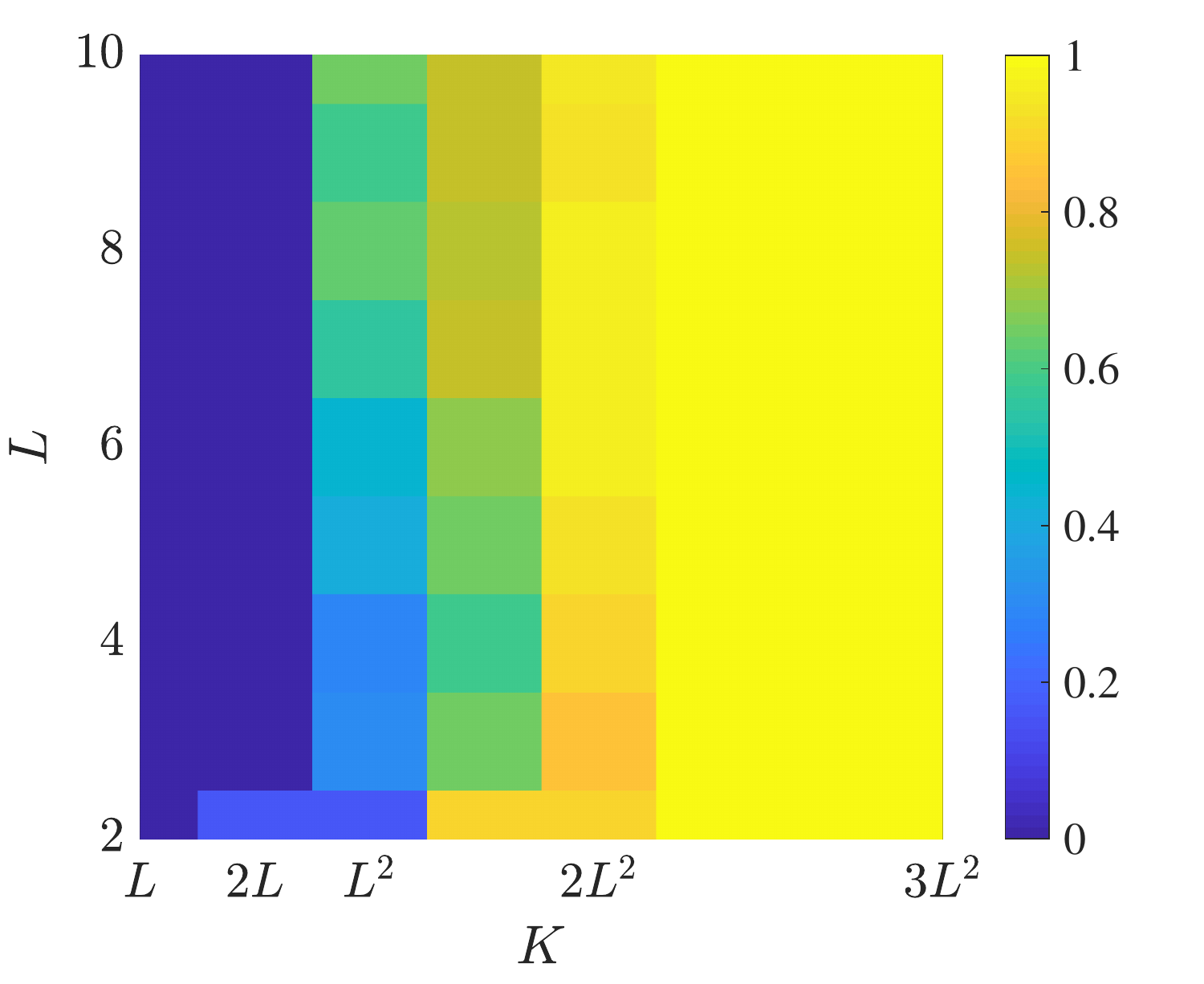}}&
\subfigure[TPI]{\includegraphics[width = 1.6in]{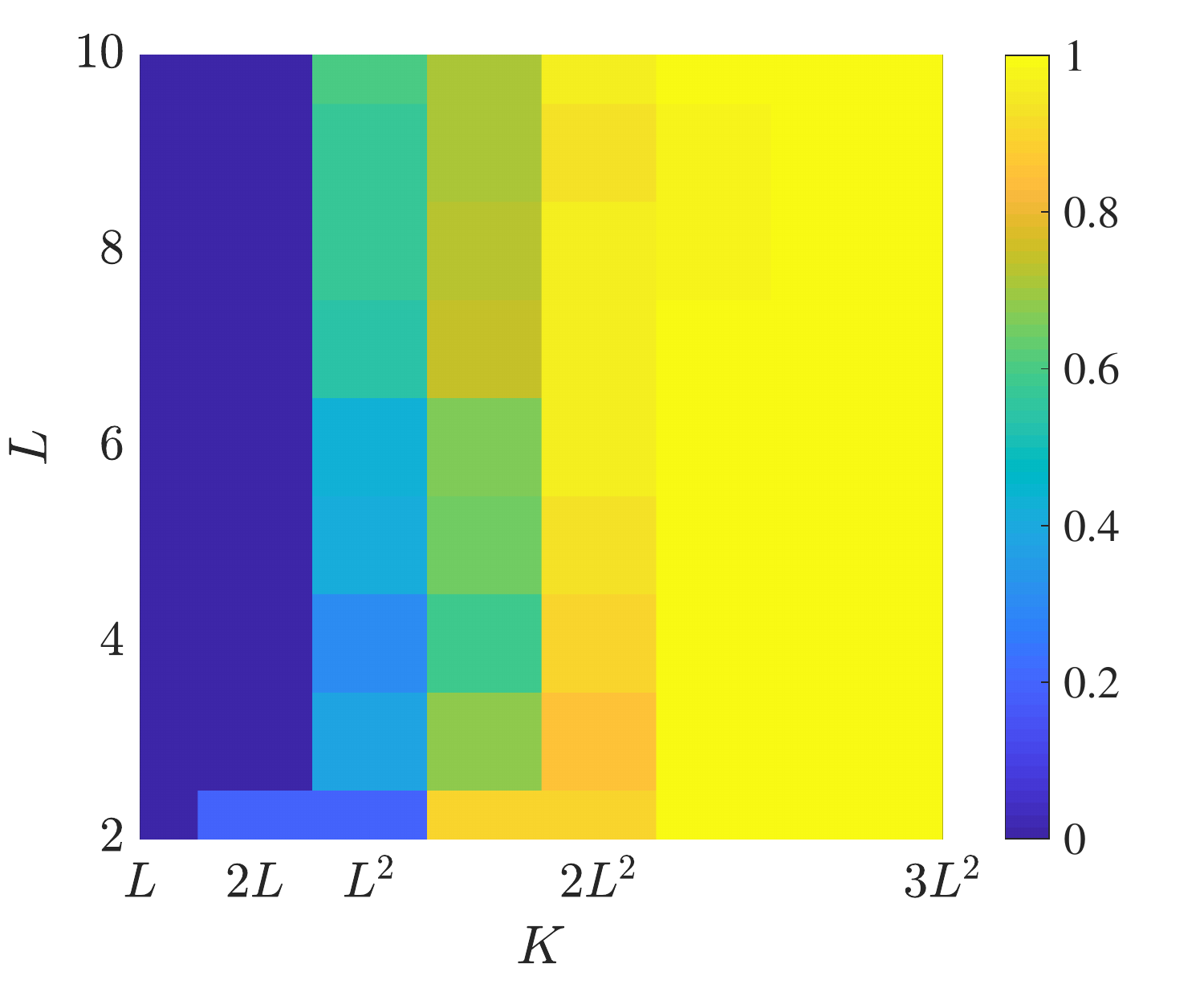}}\vspace{-.1in}
\end{tabular} 
\subfigure[BDC]{\includegraphics[width = 1.6in]{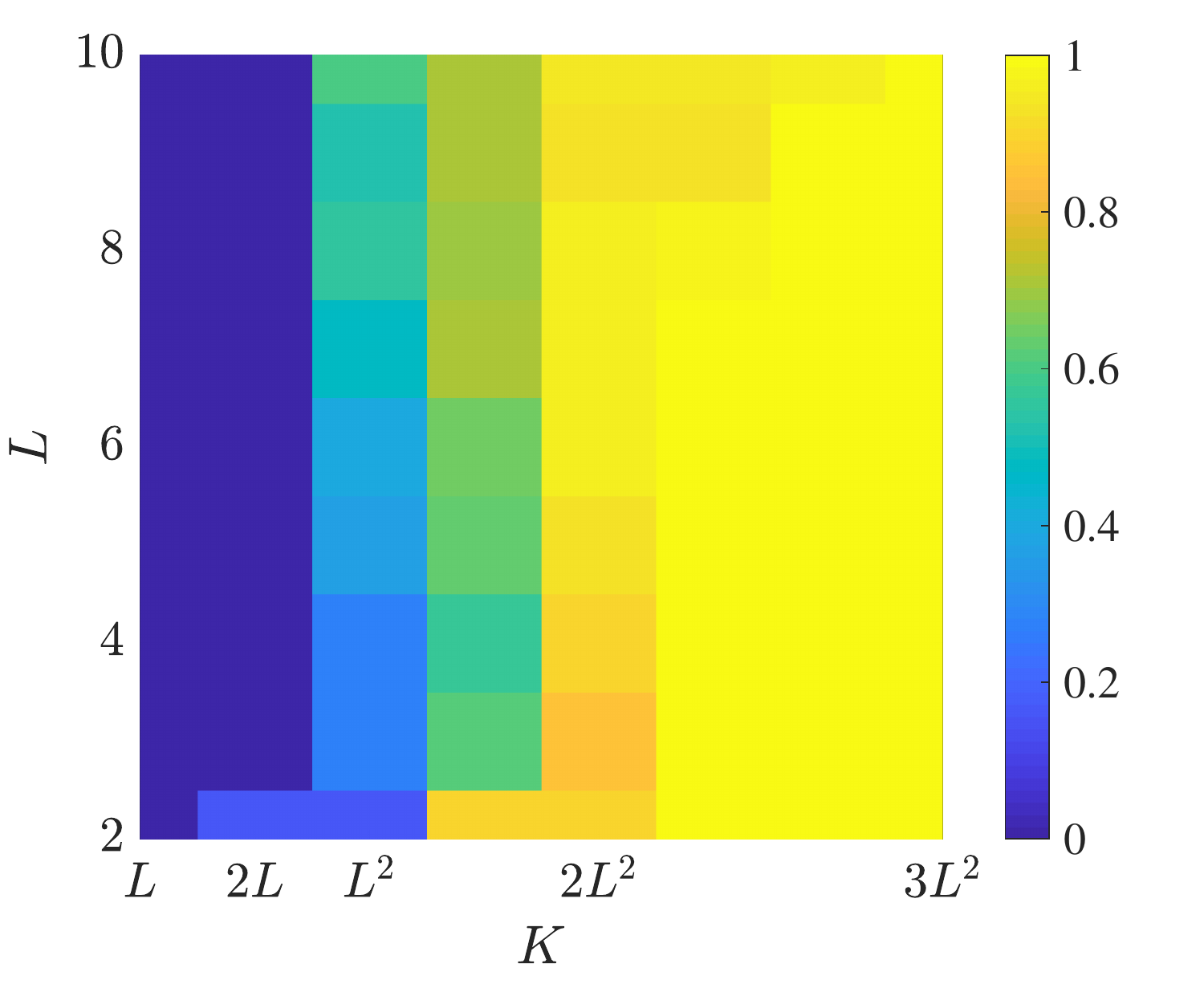}}
\end{center}
\caption{Phase transition for the identification of $L$ sparse filters in the two-channel case: $K > 2L$ is necessary and $K > 2L^2$ is sufficient. 
}
\label{fig:sparse_compare}
\end{figure}

\begin{figure}[!t]
\begin{center}
\begin{tabular}{cc}
\subfigure[Time-domain approach by Xu et al. \cite{xu_kilath}]{\includegraphics[width = 1.5in]{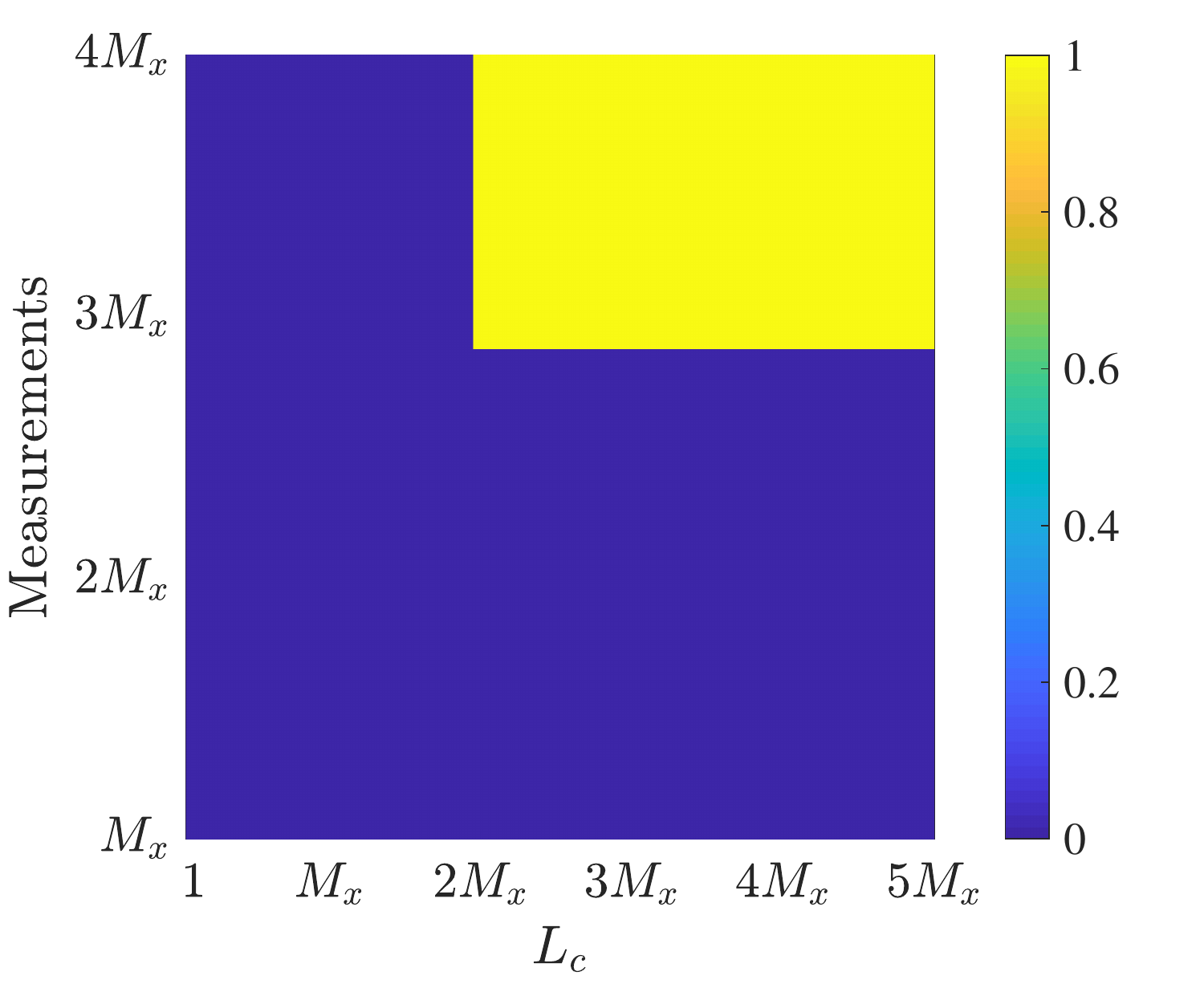}}
&
\subfigure[Frequency-domain compressive MBD]{\includegraphics[width = 1.5in]{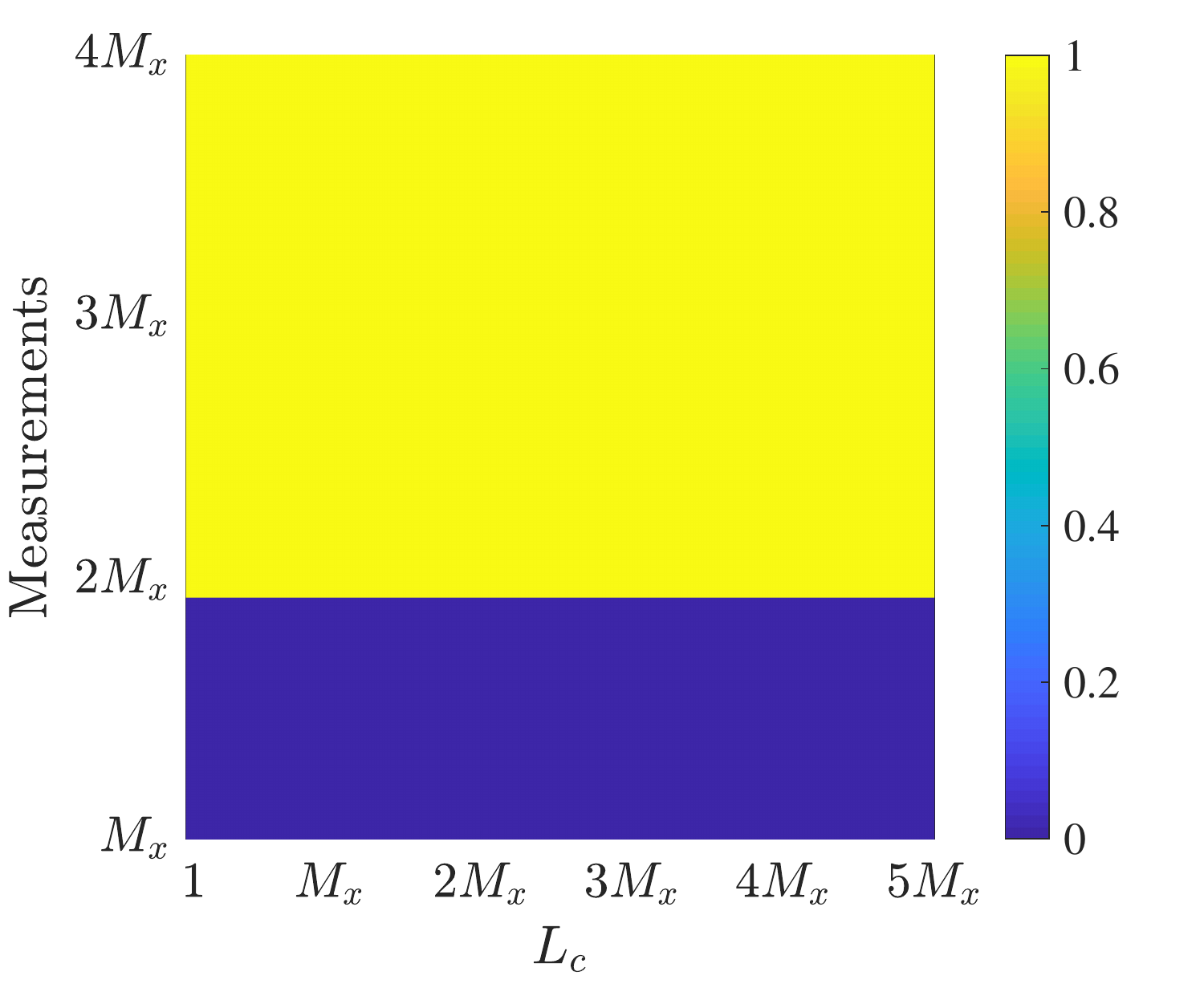}}
\end{tabular}
\end{center}
\caption{Phase transition in the non-sparse FIR case: The success rate is plotted as a function of the linear complexity $L_c$ of the source and the number of measurements $M_s-M_x$. Compressive MBD identifies the filter from at least $2M_x-1$ measurements and is independent of $L_c$. The approach by Xu et al. \cite{xu_kilath} requires at least $3M_x$ measurements and $L_c\geq 2M_x-1$.}
\label{fig:fir_compare}
\end{figure}

\subsubsection{Exhaustive Search for Compressive MBD}
The exhaustive search can be applied to either the measurement $\mathbf{Y}$ in \eqref{eq:mes_mat} or to $\mathbf{B}$ in \eqref{eq:cc_mat} to identify the unknowns. In the former case, one needs to search over three unknowns, $\mathbf{s}$, $\mathbf{x}_1$, and $\mathbf{x}_2$, whereas, the search reduces to two unknowns $\mathbf{x}_1$ and $\mathbf{x}_2$ in the latter case. Hence, we apply the exhaustive search to identify the sparse filters from $\mathbf{B}$ and compare the results with Theorem~\ref{theorem1}. 

In this experiment, the measurements are directly generated by using \eqref{eq:mes_mat} upon setting $M_x = 2L^2$, $L = 4$, and $M = 2M_x$. We choose $\mathcal{K}$ as $\{1, 2, \dots, K\}$. We assume that the full Fourier measurements of the source consist of a sum of two Gaussian pulses with their amplitudes, means, and variances given by the triplets: $(4, M/2, 0.001)$ and $(1, 2M/3, 0.01)$. We choose the sum of Gaussian pulses to make sure that the source spectrum is not flat and non-vanishing. The experiment can also be performed for any other choices of the source. The support and the coefficients of the filters are generated randomly. Specifically, for each filter $L$ non-zero values are chosen uniformly at random over the set $[M_x]$ and their amplitudes are chosen uniformly at random between 1 and 2. 
Out of the available $M$ Fourier measurements, we use $K \leq M$ to identify the filters. If the filters are uniquely identifiable up to scaling and shift of the original filters, we consider the experiment to be successful or else we  assume that the experiment has failed for that particular $K$. 

Figure~\ref{fig:Esearch} shows the success rate averaged over 200 independent experiments. Interestingly, we observe that $2L$ Fourier samples are necessary and sufficient to identify the filters.      

\subsubsection{TPI and BDC to Solve Compressive MBD}
In this section, we compare the performance of the TPI and BDC approaches to solve compressive MBD. We compare the performance of these two algorithms to that by OMP in the non-blind case where the source $s$ is known. We call this method non-blind OMP (NB-OMP). 

We use the same simulation settings as in the exhaustive search experiment. Let $\mathbf{\widetilde{X}}$ is an estimate of $\mathbf{X}$, then the filters are assumed to be identified if the normalized mean-squared error (MSE) is less than $-50$ dB, that is, if  $20\log ( \|\mathbf{X}-\mathbf{\widetilde{X}}\|_2/\|\mathbf{X}\|_2) \leq -50$. Figure~\ref{fig:L4_comp} shows the average success rate computed over 200 independent realizations of the sparse filters for $L=4$. Performances of both TPI and BDC methods follow closely that of NB-OMP. We observe that the success rate gradually increases for $K \geq 2L$ up to $K = 2L^2$. For $K>2L^2$, both BDC and TPI algorithms always identify the filters uniquely. The curve shows that $2L$ measurements are necessary and $2L^2$ compressive MBD measurements are sufficient to uniquely identify the sparse filters. 
 
Figure~\ref{fig:sparse_compare} shows the success rate for different values of sparsity levels. We note that for $L=2$, for some realizations of the filters, the algorithms are able to identify the filters for $K = 2L$. Except for the BDC method in the case of $L>8$, we observe that $K>2L^2$ measurements are sufficient. The results also show that $2L$ measurements are necessary.

\subsection{Comparison of Non-Sparse FIR MBD}
Next, we show simulation results for the non-sparse case.  As the present frequency-domain approach to derive the identifiability results and the time-domain approach taken by Xu et al. \cite{xu_kilath} are based on cross-correlation method we compare these two techniques. In both approaches, the solution is computed by solving a homogeneous equation. In our frequency-domain settings, the solution is obtained by solving the homogeneous equation \eqref{eq:cc_mat} via the optimization problem \eqref{eq:tpi_opt} without imposing the sparsity assumption. The solution is given as the eigenvector corresponding to the minimum eigenvalue of the matrix $\mathbf{B}^{\mathrm{H}}\mathbf{B}$. A similar technique is applied in Xu et al. by replacing the matrix $\mathbf{B}$ by a data matrix that consists of time-domain samples (cf. equations (7) and (20) in \cite{xu_kilath}). The identifiability results in the time-domain approach strongly depend on the linear complexity $L_c$ of the source and on the length $M_x$ of the source compared with the length of the filters $M_x$. We claim that our results do not depend on these parameters, but depend only on the number of frequency-measurements and non-vanishing property of the source at those frequencies. To justify our claim, we compare the MSE in estimating the filters by both approaches as a function of $L_c$ and the number of measurements. To have a fair comparison, we use the same number of measurements for each experiment in both the methods. In \cite{xu_kilath}, the number of measurements is given by $M_s-M_x$ by assuming that $M_s>M_x$. In the proposed frequency-domain approach, we choose $\mathcal{K}$ as set of $K = M_s-M_x$ consecutive integers. In the simulations, the source sequence $s$ of length $M_s$ with a linear complexity $L_c$ is generated using the following model: 
\begin{align}
  s[m] = \displaystyle \sum_{\ell =1}^{L_c}c_\ell r_\ell^m \quad \text{for} \quad m = 0, 1, \dots, M_s-1.
\end{align}
In the simulations, both $c_\ell$ and $r_\ell$ as well as the filters are generated randomly. 

In both approaches, we assume that the filters are uniquely identifiable if the normalized MSE is less than $-50$ dB. Fig.~\ref{fig:fir_compare} shows the success rate over 200 independent realizations of the source and filers for Xu's \cite{xu_kilath} and the proposed methods. We observe that Xu's method is successful if $L_c \geq 2M_x-1$ and $M_s\geq 4M_x$, whereas our approach identifies the filters as long as $K = M_s-M_x > 2M_x-1$ and the results are independent of the linear complexity of the source.

\section{Proof of Lemma \ref{lemma1}}
\label{sec:proof}
Theorem~\ref{theorem1} states that $(\{S(e^{\mathrm{j}\omega_0k})\}_{k \in \mathcal{K}}, x_1, x_2)$ is the unique solution to \eqref{eq:optNEW} up to the fundamental ambiguity in \eqref{eq:ambiguitu_class}. Let $(\{\hat{S}(e^{\mathrm{j}\omega_0k})\}_{k \in \mathcal{K}}, \hat{x}_1, \hat{x}_2)$ be another solution to \eqref{eq:optNEW}. Then there exist $\alpha \neq 0$ and $m_0 \in \mathbb{Z}$ such that
\begin{equation}
\label{eq:eqc1}
S(e^{\mathrm{j}\omega_0k}) = \alpha e^{\mathrm{j}\omega_0k m_0} \hat{S}(e^{\mathrm{j}\omega_0k}), \quad \forall k \in \mathcal{K},
\end{equation}
and
\begin{equation}
\label{eq:eqc2}
X_n(z) = \alpha^{-1} z^{-m_0} \hat{X}_n(z), \quad n = 1,2,
\end{equation}
if $|\mathcal{K}| \geq 2L^2$.
	
\emph{Sufficiency part: }
Since both $(\{S(e^{\mathrm{j}\omega_0k})\}_{k \in \mathcal{K}}, x_1, x_2)$ and $(\{\hat{S}(e^{\mathrm{j}\omega_0k})\}_{k \in \mathcal{K}}, \hat{x}_1, \hat{x}_2)$ are solutions to \eqref{eq:optNEW}, we have
\begin{align}
\hspace{-.09in}Y_n(e^{\mathrm{j}\omega_0k}) = S(e^{\mathrm{j}\omega_0k}) X_n(e^{\mathrm{j}\omega_0k}) = \hat{S}(e^{\mathrm{j}\omega_0k}) \hat{X}_n(e^{\mathrm{j}\omega_0k}), 
\label{eq:mes}
\end{align}
for $\quad n = 1, 2,$ and $k \in \mathcal{K}$. We have assumed that $\{S(e^{\mathrm{j}\omega_0k})\}_{k \in \mathcal{K}}$ and $\{\hat{S}(e^{\mathrm{j}\omega_0k})\}_{k \in \mathcal{K}}$ are nonzero. Therefore it follows from \eqref{eq:mes} that
\begin{equation}
\label{eq:partial_convNEW}
X_1(e^{\mathrm{j}\omega_0k}) \hat{X}_2(e^{\mathrm{j}\omega_0k}) = \hat{X}_1(e^{\mathrm{j}\omega_0k}) X_2(e^{\mathrm{j}\omega_0k}), \quad k \in \mathcal{K}.
\end{equation}
Let $Q(e^{\mathrm{j}\omega_0k})$ denote the sample of DTFT of 
\begin{align}
q = x_1 * \hat{x}_2 - x_2 * \hat{x}_1
\label{eq:q1}
\end{align}
at frequency $k\omega_0$. Then by the convolution theorem and the linearity of DTFT, the identity in \eqref{eq:partial_convNEW} is equivalently rewritten as
\begin{align}
Q(e^{\mathrm{j}\omega_0k}) = 0, \quad k \in \mathcal{K}.
\label{eq:coincideDnew}
\end{align}
	
Let $\mathbf{A} \in \mathbb{C}^{|\mathcal{K}| \times 2M_x-1}$ denote a Vandermonde matrix with its $(k,m)$th entry given as $e^{\mathrm{j} (k-1)(m-1)\omega_0}$. Then we can rewrite \eqref{eq:coincideDnew} as
\begin{align}
\mathbf{Aq} = \mathbf{0},
\label{eq:Aq}
\end{align}
where $\mathbf{q} = [q[0], q[1], \dots, q[2M_x-2]]^{\mathrm{T}} \in \mathbb{C}^{2M_x-1}$. Since $\mathbf{A}$ is a Vandermonde matrix constructed by $2M_x-1$ distinct generators $\{e^{\mathrm{j}m\omega_0}\}_{m \in [2M_x]}$ and $\mathcal{K}$ is a universal set, $\mathbf{A}$ has full spark. Therefore, if $|\mathcal{K}| \geq |\text{supp}\{q\}|$, then \eqref{eq:Aq} implies $\mathbf{q}=\mathbf{0}$. Furthermore, since $\text{supp}\{q\} \subset [2M_x]$, it follows that $q = 0$, that is,
\begin{align}
x_1 * \hat{x}_2 =  x_2 * \hat{x}_1, 
\label{eq:ratio2new}
\end{align}
which implies via the $z$-transform that 
\begin{align}
X_1(z)/X_2(z)=\hat{X}_1(z)/\hat{X}_2(z).
\label{eq:ratio4}
\end{align}
	
For brevity, we introduce the following notation. For an FIR sequence $x$, let $\mathcal{Z}_x$ denote the set of the zeros of its $z$-transform $X(z)$.
	
Since $X_1(z)$ and $X_2(z)$ do not share any common zeros, we obtain
\begin{align}
\mathcal{Z}_{x_1}\subseteq \mathcal{Z}_{\hat{x}_1}.
\label{eq:zeros}
\end{align}
Hence, there is a polynomial $H(z)$ that satisfies
\[
\hat{X}_1(z)= X_1(z)H(z).
\]
Then it follows from \eqref{eq:ratio4} that
\[
\hat{X}_2(z)= X_2(z)H(z).
\]
Since $\hat{X}_1(z)$ and $\hat{X}_2(z)$ do not share any common zeros except at 
$z=0$, we conclude that $H(z) = \alpha z^{m_0}$ for some $m_0 \in \mathbb{N}$ and $\alpha \in \mathbb{C} \setminus \{0\}$.
It remains to show that $\|q\|_0 = |\text{supp}\{q\}| \leq2 L^2$.
Without any assumption on the support structure of the filters the support of sequence $q$ depends only on $L$. Since $L<\sqrt{M_x}$, we have $\|x_1* \hat{x}_2\|_0 =\|x_2* \hat{x}_1\|_0 \leq L^2$. Hence
\begin{align}
\|q\|_0=\|x_1 * \hat{x}_2 - x_2 * \hat{x}_1\|_0 \leq 2L^2.\label{eq:l0new}
\end{align}
	
\emph{Necessity part: } We show that if $|\mathcal{K}| < 2L$, then there exist distinct solutions $(\{S(e^{\mathrm{j}\omega_0k})\}_{k \in \mathcal{K}}, x_1, x_2)$ and $(\{\hat{S}(e^{\mathrm{j}\omega_0k})\}_{k \in \mathcal{K}}, \hat{x}_1, \hat{x}_2)$ to \eqref{eq:optNEW} such that \eqref{eq:eqc1} and \eqref{eq:eqc2} are not satisfied for any $\alpha$ and $n_0$. It suffices to show that 
\begin{align}
q = {x}_1 * {\hat{x}}_2-{x}_2 * {\hat{x}}_1\neq 0, \quad \text{but} \quad Q(e^{\mathrm{j}\omega_0k}) = 0, \quad \forall k \in \mathcal{K}.\nonumber
\label{eq:necessary}
\end{align}
	
Before proceeding further, we define the following notations to prove the necessary part. Let $\mathbf{\bar{A}}$ denotes the $|\mathcal{K}|\times M_x$ matrix that consists of first $M_x$ columns of matrix $\mathbf{A}$. Next, let $\mathbf{x}_1$, $\mathbf{x}_2$, $\mathbf{\hat{x}}_1$, and $\mathbf{\hat{x}}_2$ denote $L$-sparse vectors in $\mathbb{C}^{M_x}$ that are constructed by considering the first $M_x$ values of the sequences $x_1$, $x_2$, $\hat{x}_1$, $\hat{x}_2$, respectively. 

Since the matrix $\mathbf{\bar{A}}$ has full spark, there exist distinct pairs of $L$-sparse vectors $(\mathbf{x}_1, \mathbf{\hat{x}}_1)$ and $(\mathbf{x}_2, \mathbf{\hat{x}}_2)$ such that $\mathbf{\bar{A}}\mathbf{x}_1 = \mathbf{\bar{A}}\mathbf{\hat{x}}_1$ and $\mathbf{\bar{A}}\mathbf{x}_2 = \mathbf{\bar{A}}\mathbf{\hat{x}}_2$ as long as $|\mathcal{K}|<2L$. Hence, we have that
\begin{align}
\mathbf{\bar{A}\mathbf{x}_1
 ./ \mathbf{\bar{A}}\mathbf{x}_2} = \mathbf{\bar{A}}\mathbf{\hat{x}}_1  ./ \mathbf{\bar{A}}\mathbf{\hat{x}}_2,
\end{align}
which is equivalent to \eqref{eq:partial_convNEW} or $Q(e^{\mathrm{j}\omega_0k}) = 0$ for all $k \in \mathcal{K}$. Since $\mathbf{x}_1 \neq \mathbf{\hat{x}}_1$ and $\mathbf{x}_2 \neq \mathbf{\hat{x}}_2$, we also have that corresponding sequence $q \neq 0$. Hence, for $|\mathcal{K}|<2L$ the problem in \eqref{eq:optNEW} does not have unique solution for all the filter pairs $(x_1, x_2)$.

\section{Conclusions}	
In this paper, we derived identifiability conditions for multichannel blind deconvolution when the filters follow a fully deterministic sparsity model and the source has finite support. We showed that when there exist at least a pair of two mutually coprime filters, it is sufficient to take $2L^2$ Fourier measurements from those channels for the unique identification of the filters. To identify the source uniquely, we derive conditions on the number of measurements and number of channels in terms of the support of the source and sparsity of the filters. The results improve upon existing MBD results both in terms of the number of measurements and the number of channels required for unique identifiability and also apply to the non-sparse settings. 
	
\ifCLASSOPTIONcaptionsoff
\newpage
\fi
	
\bibliographystyle{IEEEtran}
\bibliography{refs,refs2}
	
\end{document}